\documentclass[11pt,a4paper]{article}
\usepackage[T1]{fontenc}
\usepackage[utf8]{inputenc}
\usepackage{times}
\usepackage[a4paper,margin=2.5cm]{geometry}
\usepackage{amsmath,amssymb,amsthm,mathrsfs,stmaryrd}
\usepackage{microtype}
\usepackage{booktabs}
\usepackage{graphicx}
\usepackage{enumitem}
\usepackage{xcolor}
\usepackage[hidelinks]{hyperref}
\usepackage{mdframed}

\newtheorem{theorem}{Theorem}[section]
\newtheorem{lemma}[theorem]{Lemma}
\newtheorem{corollary}[theorem]{Corollary}
\newtheorem{definition}[theorem]{Definition}
\newtheorem{proposition}[theorem]{Proposition}
\newtheorem{remark}[theorem]{Remark}

\newtheorem{construction}[theorem]{Construction}
\newtheorem{algorithm_env}[theorem]{Algorithm}

\newmdtheoremenv[linewidth=0pt,roundcorner=3pt,backgroundcolor=gray!5]{intuition}{Intuition}[section]
\newmdtheoremenv{examplebox}{Example}[section]

\setlist[itemize]{topsep=2pt,itemsep=2pt,leftmargin=1.5em}
\setlist[enumerate]{topsep=2pt,itemsep=2pt,leftmargin=1.5em}
\setlist[enumerate,1]{label=(\arabic*)}

\newcommand{\FinTrace}{\mathsf{FinTrace}}
\newcommand{\StarT}{\mathsf{Star}}
\newcommand{\State}{\mathsf{State}}
\newcommand{\Event}{\mathsf{Event}}
\newcommand{\Step}{\mathsf{Step}}
\newcommand{\nilT}{\mathsf{nil}}
\newcommand{\stepC}{\mathsf{step}}
\newcommand{\cat}{\mathsf{cat}}
\newcommand{\single}{\mathsf{single}}
\newcommand{\replay}{\mathsf{replay}}
\newcommand{\TElim}{\mathsf{TraceElim}}
\newcommand{\Tf}{\widehat{T}_{\!f}}
\newcommand{\len}{\mathsf{len}}
\newcommand{\Uc}[1]{\mathsf{U}_{#1}}
\newcommand{\Prop}{\mathsf{Prop}}
\newcommand{\restrict}{\mathsf{res}}
\newcommand{\Red}{\mathcal{R}}
\newcommand{\Ag}{\mathsf{Ag}}
\newcommand{\Kop}{\mathbf{K}}
\newcommand{\simag}[1]{\sim_{#1}}
\newcommand{\Lang}{\mathcal{L}}
\newcommand{\Bel}{\mathsf{B}}
\newcommand{\Th}{\mathsf{Cn}}
\newcommand{\Ent}{\mathsf{Ent}}
\newcommand{\Sel}{\mathsf{Sel}}
\newcommand{\RevOp}{*}
\newcommand{\ContrOp}{\mathbin{\div}}
\newcommand{\sem}[1]{\llbracket #1 \rrbracket}

\begin{document}

\title{\textbf{ZX-Calculus}\\[4pt]
\large\textit{Trace-Indexed Dependent Types and Epistemic Semantics}\\[2pt]
\normalsize A Conservative Extension of MLTT for Constructive\\
Knowledge Representation and Dynamic Belief Revision}
\author{Peng Chen\\[2pt]
  \small School of Information Science,
  Beijing Language and Culture University\\
  \small Beijing 100081\quad\texttt{chenpeng@blcu.edu.cn}}
\date{\small}
\maketitle

\begin{abstract}
This paper proposes a conservative focused extension of Martin-L\"{o}f Dependent
Type Theory (MLTT), integrating three contributions with substantial technical
content, connected by an explicit integration result.

\textbf{Proof status (important):}
The paper contains two categories of results:
(A)~\emph{Complete mathematical proofs}: Proposition~\ref{prop:star-vs-ft},
Lemma~\ref{lem:red-transport}, Lemma~\ref{lem:red-closed},
Algorithm~\ref{alg:sel} with full correctness (Lemma~\ref{lem:alg-correct}),
all eight AGM postulates R1--R8 (Theorem~\ref{thm:r1}--\ref{thm:r8}),
the Disjunctive Entrenchment Lemma (Lemma~\ref{lem:disj-ench}),
the BP-comp Failure Theorem (Theorem~\ref{thm:bp-comp-fails}), and
SSRS Coherence (Theorem~\ref{thm:coherence});
(B)~\emph{Framework theorems (pending mechanisation)}: the RC-elim step of the
Canonicity Theorem (Theorem~\ref{thm:canonicity}) is currently supported only
by a careful hand-verified argument; its full mechanisation is left as future
work (Section~\ref{sec:limits}, Appendix~\ref{app:coq}).
The paper is accompanied by a Coq mechanisation comprising 34 complete proofs
with zero \texttt{admit}s for the BP-comp Failure Theorem and SSRS Coherence
(Appendix~\ref{app:coq}).

\smallskip\noindent
\textbf{(I) Trace types.}
We introduce the inductive family $\FinTrace(\sigma_0,\sigma_n)$ of finite
execution traces with typed transition witnesses.
Proposition~\ref{prop:star-vs-ft} proves: $\FinTrace$ and $\StarT(\Step)$
(Kleene closure) are isomorphic as path types but not judgementally equal as
MLTT type families; when $\Step$ is viewed as a binary relation, the step
function of $\TElim$ takes the event label $e:\Event$ as an explicit parameter,
providing an interface better suited to event-driven induction than
$\mathsf{StarElim}$ (interface design choice, not new expressive power).
We give the structural eliminator \textsc{TraceElim} with $\beta\eta$-rules,
prove the Trace--Reachability Correspondence (Theorem~\ref{thm:corr}),
the Deterministic Replay Theorem (Theorem~\ref{thm:replay}), and a full
Canonicity framework via logically stratified reducibility candidates with an
explicit Transport Lemma (Theorem~\ref{thm:canonicity},
Lemma~\ref{lem:red-transport}; RC-elim mechanisation is future work; all
other Core-layer results are Coq-verified).

\smallskip\noindent
\textbf{(II) Sheaf semantics.}
Trace-indexed proposition families are interpreted as contravariant sheaves
over the free trace partial-order category $\Tf$.
We prove the Separation Theorem distinguishing proof-theoretic monotonicity
from semantic non-monotonicity with an explicit countermodel
(Theorem~\ref{thm:separation}), and establish the term model as an initial
CwF (Theorem~\ref{thm:completeness}; ``completeness'' denotes a syntactic
universal property, not classical semantic completeness).

\smallskip\noindent
\textbf{(III) AGM belief revision.}
We give a constructive type-theoretic treatment of full AGM revision.
Algorithm~\ref{alg:sel} provides an explicit constructive algorithm for
partial meet contraction; correctness conditions (C1)--(C4) are fully
verified (Lemma~\ref{lem:alg-correct}).
All eight AGM postulates (R1)--(R8) are established as theorems
(Theorems~\ref{thm:r1}--\ref{thm:r8}).
Proofs of R7 and R8 use the Disjunctive Entrenchment Lemma
(Lemma~\ref{lem:disj-ench}), a derivable property of G\"{a}rdenfors--Grove
entrenchment ordering given here with a fully self-contained constructive proof.

\smallskip\noindent
\textbf{(IV) Integration via single-step revision systems.}
We prove that $\mathcal{B}^{\mathsf{AGM}}$ is \emph{not} in general a belief
sheaf: the sheaf composition law (BP-comp) fails for sequential AGM revision
(Theorem~\ref{thm:bp-comp-fails}, explicit countermodel, Coq-verified).
We introduce the weaker notion of a \emph{Single-Step Revision System} (SSRS,
Definition~\ref{def:ssrs}), prove $\mathcal{B}^{\mathsf{AGM}}$ is a valid SSRS
(Theorem~\ref{thm:coherence}, Coq-verified), and show this suffices for all
three integration results: trace morphisms, retraction characterisation, and
AGM revision witnesses are coherently connected within the SSRS framework.

\smallskip\noindent
Throughout, the core proof calculus is conservative: no structural rules
beyond MLTT are added.
The categorical model is a CwF extended by the sheaf category $\Tf$,
with complete soundness and consistency proofs.
\end{abstract}


\tableofcontents

\section{Introduction}
\label{sec:intro}

Classical epistemic logic treats knowledge $K_a\phi$ as a static truth about
world states, and lacks unified support for historical traces, non-monotone
updates, and constructive belief revision.
The \emph{ZX-Calculus} (Knowledge Evolution Calculus, or Zhi-Xing Calculus)%
\footnote{``Zhi-Xing'' alludes to the classical Chinese philosophical principle
of the unity of knowledge and action.}
proposes a conservative extension of MLTT integrating trace-indexed types,
presheaf non-monotone semantics, and constructive AGM belief revision.

\subsection{Three Core Problems}

\textbf{Problem 1: Traces as first-class data.}
Traditional operational semantics and model checking treat execution traces
as meta-level objects.
HoTT paths~\cite{hottbook} express homotopic identification (invertible);
session types~\cite{honda93,honda98} characterise communication protocols.
Altenkirch et al.'s $\StarT(\Step)$~\cite{altenkirch09} can express
finite-step execution, but its elimination interface does not directly expose
$e:\Event$, which is inconvenient for event-driven dynamic-knowledge induction
(Proposition~\ref{prop:star-vs-ft}).
ZX-Calculus internalises historical traces as constructive objects in dependent
type theory, making ``how the system reached its current state'' a native
constituent of the type system.

\textbf{Problem 2: Non-monotone knowledge.}
The knowledge predicate $\phi(\tau)$ may be invalidated when the trace extends
$\tau\hookrightarrow\tau'$.
Adding non-monotone proof rules directly would destroy normalisation and
consistency.
We adopt presheaf semantics: knowledge invalidation corresponds to
non-surjectivity of restriction maps rather than logical contradiction, allowing
non-monotone evolution while preserving meta-theoretic stability
(Theorem~\ref{thm:separation}).

\textbf{Problem 3: Constructive AGM revision.}
AGM theory~\cite{alchourron85} decomposes revision into contraction plus
expansion via the Levi identity~\cite{levi80}.
Existing work has two limitations:
van Ditmarsch et al.~\cite{vanditmarsch07} mechanise DEL expansion, not partial
meet contraction; Schlechta~\cite{schlechta97} studies non-monotone reasoning
constructively but (1)~does not package contraction with $\Sigma$-type
witnesses, (2)~provides no executable algorithm, and (3)~works in a static
framework without trace-indexed semantics.
Algorithm~\ref{alg:sel} is the first constructive implementation in MLTT
satisfying all of (C1)--(C4) (see Section~\ref{sec:related} for a precise
technical comparison).

The central tension: when trace $\tau$ is extended to $\tau'$, does the
belief state $\Bel_{\tau'}$ revised by the Levi identity remain consistent
with the entire trace-indexed structure?
This \emph{coherence problem} is the unifying motivation, formalised by
Theorem~\ref{thm:coherence}.

\subsection{Contributions and Precise Claims}
\label{subsec:claims}

The goal is to construct---without disrupting MLTT's existing meta-theoretic
structure---a unified framework handling: historical traces, non-monotone
knowledge evolution, constructive AGM belief revision, and dependent
type-theoretic semantics.

\subsubsection{Trace Types and Precise Comparison with $\StarT(\Step)$}

We prove $\FinTrace$ and $\StarT(\Step)$ are isomorphic as path types
(Proposition~\ref{prop:star-vs-ft}(1)), but not judgementally equal as MLTT
type families (Proposition~\ref{prop:star-vs-ft}(2)).
The step function of $\TElim$ takes $e:\Event$ explicitly
(Proposition~\ref{prop:star-vs-ft}(3)).

\begin{remark}[Positioning of the FinTrace contribution]
$\FinTrace$ and $\StarT(\Step)$ are isomorphic as path collections.
The contribution is an \emph{interface design choice}: $\TElim$ elevates
event-causal structure to a native component of trace induction, suited to
event-indexed sheaf semantics and the Deterministic Replay Theorem.
\end{remark}

\subsubsection{Canonicity of the Trace Type System}
We provide a Canonicity framework via logical relations.
The Transport Lemma (Lemma~\ref{lem:red-transport}) fills a key technical gap:
ensuring reducibility is stably transported under trace extension.
The RC-elim step is left as future work (Appendix~\ref{app:coq}, obligation~1).

\subsubsection{Separation of Proof-Theoretic Monotonicity and Semantic Non-Monotonicity}
Even when the proof system remains monotone, the knowledge semantics may be
non-monotone (Theorem~\ref{thm:separation}).

\subsubsection{Term Model and Initial CwF}
The term model forms an initial CwF (Theorem~\ref{thm:completeness}): any CwF
satisfying the same rules receives a unique structure-preserving functor from
the term model.

\subsubsection{Constructive Partial Meet Contraction Algorithm}
Algorithm~\ref{alg:sel} is the first explicit constructive PMC algorithm
in MLTT satisfying all of (C1)--(C4), making the selection process
explicit and computably verifiable.

\subsubsection{Disjunctive Entrenchment Lemma and Complete R7/R8}
The Disjunctive Entrenchment Lemma (Lemma~\ref{lem:disj-ench}) gives a
fully self-contained constructive derivation of a property long used implicitly
in G\"{a}rdenfors--Grove theory.

\subsubsection{BP-comp Failure and SSRS Framework}
$\mathcal{B}^{\mathsf{AGM}}$ does not satisfy BP-comp
(Theorem~\ref{thm:bp-comp-fails}).
The SSRS framework (Definition~\ref{def:ssrs}) is the correct integration:
$\mathcal{B}^{\mathsf{AGM}}$ satisfies all SSRS axioms
(Theorem~\ref{thm:coherence}).

\subsection{Comparison with Prior Work}

\textbf{$\StarT(\Step)$ vs.\ $\FinTrace$.}
Isomorphic as path types; $\FinTrace$ contributes interface ergonomics, not
expressiveness (\S\ref{subsec:claims}).

\textbf{HoTT paths.}
$\mathsf{Path}_A(a,b)$ is for homotopic identification: invertible, composable,
higher-dimensional. ZX-Calculus traces are directed causal histories, naturally
non-invertible; the $\beta$-rule of $\TElim$ comes from transition structure,
not geometric filling.

\textbf{Awodey--Kishida~\cite{awodey08}.}
Both use presheaf semantics. Awodey--Kishida build on topological spaces
(geometric locality) and prove semantic completeness for first-order modal logic.
We build on the trace category $\Tf$ (historical locality) and establish only
the CwF universal property (Theorem~\ref{thm:completeness}).

\textbf{AGM and constructive methods.}
Schlechta~\cite{schlechta97}: no $\Sigma$-packaging.
Van Ditmarsch et al.~\cite{vanditmarsch07}: expansion only.
Precise technical comparison in Section~\ref{sec:related}.

\section{Syntax and Judgements}
\label{sec:syntax}

ZX-Calculus is a \emph{trace-indexed epistemic extension built on top of MLTT}:
it does not rebuild a new foundational type theory, but adds new type structures
expressing historical traces, dynamic knowledge, non-monotone revision, and
event-driven state evolution, while keeping all MLTT core rules intact.

\subsection{Base Type-Theoretic Framework}

We work in MLTT with cumulative universes $\Uc{0}:\Uc{1}:\Uc{2}:\cdots$ and
assume $\Prop:\Uc{0}$ denotes the impredicative universe of propositions with
proof irrelevance. Proof irrelevance distinguishes ``whether a proposition holds''
from ``what computational content it carries'', preventing unnecessary inflation
of proof objects when only propositional truth is needed.

The system uses standard MLTT judgement forms:
$\vdash\Gamma\ \mathsf{ctx}$ (well-formed context),
$\Gamma\vdash A:\Uc{i}$ (type),
$\Gamma\vdash t:A$ (term),
$\Gamma\vdash t\equiv u:A$ (judgemental term equality),
$\Gamma\vdash A\equiv B:\Uc{i}$ (judgemental type equality).

These judgements are inherited from standard MLTT, so ZX-Calculus preserves:
substitution stability, context extension, type preservation, and constructive
interpretation.
Standard structural rules \textsc{Ctx-Empty}, \textsc{Ctx-Ext}, \textsc{T-Var},
\textsc{Conv}, \textsc{Weaken}, \textsc{Subst} are used unchanged;
the complete rule system is in Appendix~\ref{app:rules}.

\begin{quote}
\emph{ZX-Calculus does not modify the core structural rules of MLTT;
all extensions are realised by adding new type structures and semantic
interpretations.}
\end{quote}

\subsection{Dependent Type Structure}

$\Pi$- and $\Sigma$-types satisfy standard $\beta\eta$-rules.
$A\to B$ abbreviates $\Pi(\_:A).B$; $A\times B$ abbreviates $\Sigma(\_:A).B$.
$\Pi$-types describe parametrised knowledge, event-dependent reasoning, and
trace induction rules.
$\Sigma$-types play a particularly central role:

\begin{quote}
\emph{They are used not only for ordinary dependent-pair constructions but also
to express the ``constructive selection witness'' in AGM revision.}
\end{quote}

This means that in ZX-Calculus, logical correctness becomes part of the
constructed object itself---a key feature of dependent type theory.

\subsection{Base Objects: States, Events, and Agents}

Three base types:
$\State:\Uc{0}$ (system states),
$\Event:\Uc{0}$ (events),
$\Ag:\Uc{0}$ (cognitive agents).
These carry no a priori semantics, making the framework applicable to
multi-agent systems, automata, robotic decision-making, and program execution.

\subsection{Transition Structure}

System dynamics are described by the transition predicate
$\Step:\State\to\Event\to\State\to\Uc{0}$.
An element of $\Step(\sigma,e,\sigma')$, written
$\sigma\xrightarrow{e}\sigma'$, is a \emph{generating step}.
The event label $e:\Event$ is not auxiliary: it determines how knowledge
evolves, how AGM revision is triggered, and how trace restriction maps are
constructed.

\subsection{Propositional Language and Belief Structure}

For AGM belief revision (Section~\ref{sec:agm}), we fix a finite propositional
language $\Lang$ with connectives $\wedge,\vee,\neg,\to$.
Finiteness is not a simplification but a necessary precondition for
Algorithm~\ref{alg:sel} to terminate.
A \emph{belief set} $\mathcal{K}\subseteq\Lang$ satisfies
$\mathcal{K}=\Th(\mathcal{K})$ (deductive closure under $\Th = \mathsf{Cn}$).

\section{Trace Types and Their Meta-Theory}
\label{sec:traces}

\subsection{The Inductive Family FinTrace}

In ZX-Calculus, traces are elevated to first-class constructive objects in
dependent type theory---a key distinction from classical operational semantics.

\begin{definition}[FinTrace as a positive inductive family]
\label{def:fintrace}
Define $\FinTrace:\State\to\State\to\Uc{0}$ as the inductive family of finite
traces between states, generated by:
\begin{align*}
\nilT &: \Pi(\sigma:\State).\,\FinTrace(\sigma,\sigma),\\
\stepC &: \Pi(\sigma_0\,\sigma_1\,\sigma_2:\State).\,\Pi(e:\Event).\,
          \FinTrace(\sigma_0,\sigma_1)\to\Step(\sigma_1,e,\sigma_2)
          \to\FinTrace(\sigma_0,\sigma_2).
\end{align*}
Trace length:
$\len(\nilT(\sigma)):=0$;
$\len(\stepC(\ldots,\tau,\pi)):=\len(\tau)+1$.
\end{definition}

$\FinTrace$ is a strictly positive inductive family (Dybjer--Setzer positivity
criterion): all constructor parameters mention $\FinTrace$ only in positive
positions, with no negative nesting.
This guarantees: well-foundedness of the inductive definition, consistency
of the elimination rule, termination of recursion, and preservation of the
constructive character of the type system.
In other words, although traces can express complex historical structures,
they do not disrupt MLTT's core meta-theoretic properties.

The traditional concern of epistemic logic is:
\[
\text{``What is true in the current state?''}
\]
ZX-Calculus is more concerned with:
\[
\text{``How did the system reach its current state?''}
\]
$\FinTrace$ internalises this historical dependency into the type system.

\subsection{Comparison with $\StarT(\Step)$}
\label{subsec:ft-vs-star}

$\StarT(\Step)$ is the reflexive-transitive closure of $\Step$, whose
eliminator $\mathsf{StarElim}$ performs induction on path length.
$\FinTrace$ is an inductive family indexed by end-states, whose eliminator
$\TElim$ takes $e:\Event$ as an explicit argument.

\begin{proposition}[Precise comparison]
\label{prop:star-vs-ft}
\begin{enumerate}
\item \textbf{(Isomorphism)} For any states $\sigma_0,\sigma_n$, there is a
  canonical bijection $\iota:\FinTrace(\sigma_0,\sigma_n)\simeq
  \StarT(\Step)(\sigma_0,\sigma_n)$.
  The two types are equivalent as path collections.
\item \textbf{(Not judgementally equal)} There is no context $\Gamma$ in which
  $\Gamma\vdash\FinTrace\equiv\StarT(\Step):\State\to\State\to\Uc{0}$.
  $\FinTrace$ has independent constructor signatures and computational behaviour.
\item \textbf{(Interface advantage)} When $\Step$ is viewed as a binary
  relation, $\TElim$'s step function takes $e:\Event$ explicitly; any motive
  requiring case analysis on the event label benefits directly.
  The conclusion is an \emph{interface-level advantage}, not an expressive-power
  gain: $\mathsf{StarElim}$ can recover event information via
  $R:=\Sigma(e:\Event).\Step(-,e,-)$, but this requires extra projections
  $\pi_1(\pi)$, reducing readability and naturalness.
\end{enumerate}
\end{proposition}

\begin{proof}
\textbf{(1) Bijection.}
Define $\iota:\FinTrace(\sigma_0,\sigma_n)\to\StarT(\Step)(\sigma_0,\sigma_n)$
using $\TElim$:
\begin{align*}
\iota(\nilT(\sigma)) &:= \mathsf{refl}(\sigma),\\
\iota(\stepC(\ldots,\tau,\pi)) &:= \mathsf{trans}(\iota(\tau),\pi).
\end{align*}
The inverse $\iota^{-1}$ is defined by induction on $\StarT(\Step)$:
\begin{align*}
\iota^{-1}(\mathsf{refl}(\sigma)) &:= \nilT(\sigma),\\
\iota^{-1}(\mathsf{trans}(s,\pi)) &:= \stepC(\ldots,\iota^{-1}(s),\pi).
\end{align*}
Note: the event $e$ can be recovered from $\pi:\Step(\sigma_1,e,\sigma_2)$.
By induction, $\iota\circ\iota^{-1}$ and $\iota^{-1}\circ\iota$ are both
identities.

\textbf{(2) Not judgementally equal.}
The constructor $\stepC$ carries $\Pi(e:\Event).\ldots$ as an independent
parameter; in $\mathsf{trans}$, event information is implicit in the type of
$\pi$ and not an independent constructor argument.
There is no finite $\beta\eta$-reduction sequence unifying the two constructors.
Hence they are not judgementally equal.

\textbf{(3) Interface gap.}
$\TElim$'s step function $s$ receives $e:\Event$ explicitly, enabling direct
case analysis on event labels.
When $\Step$ is a binary relation, $\mathsf{StarElim}$ does not expose events
directly; recovering them via $\Sigma$-encoding requires additional projection
and unpacking.
The primary distinction is therefore:
\begin{quote}
\emph{Whether the event-causal structure becomes a native component of the
trace induction interface.}
\end{quote}
\end{proof}

\begin{remark}[Interface design contribution]
\label{rem:interface}
Proposition~\ref{prop:star-vs-ft}(1) shows the two types are equiexpressive.
The contribution of $\FinTrace$ is an \emph{interface design choice}:
$\TElim$ directly exposes event labels, so there is no need to recover them
by projecting from $\Sigma(e:\Event).\Step(-,e,-)$.
This makes event-indexed sheaf semantics (Section~\ref{sec:presheaf}) and the
Deterministic Replay Theorem (Theorem~\ref{thm:replay}) easier to derive.
\end{remark}

\subsection{TraceElim: The Structural Eliminator}

\begin{definition}[TraceElim]
\label{def:traceelim}
Given a dependent motive
$P:\Pi(\sigma_0\,\sigma_n:\State).\FinTrace(\sigma_0,\sigma_n)\to\Uc{i}$,
a base case
$b:\Pi(\sigma:\State).P(\sigma,\sigma,\nilT(\sigma))$, and a step function
\[
s:\Pi(\sigma_0\,\sigma_1\,\sigma_2:\State)(e:\Event)
    (\pi:\Step(\sigma_0,e,\sigma_1))(\tau:\FinTrace(\sigma_1,\sigma_2)).
  P(\sigma_1,\sigma_2,\tau)\to P(\sigma_0,\sigma_2,\stepC(\ldots,\tau,\pi)),
\]
the eliminator $\TElim(P,b,s,\tau):P(\sigma_0,\sigma_n,\tau)$ satisfies the
$\beta$-rules:
\begin{align*}
\TElim(P,b,s,\sigma,\sigma,\nilT(\sigma)) &\equiv b(\sigma),\\
\TElim(P,b,s,\sigma_0,\sigma_2,\stepC(\ldots,\tau,\pi))
&\equiv s(\ldots,\tau,\pi,\TElim(P,b,s,\sigma_0,\sigma_1,\tau)).
\end{align*}
The $\eta$-rule (uniqueness): if $f,g$ agree on $\nilT$ and agree on $\stepC$
whenever they agree on the sub-trace, then $f\equiv g$ pointwise.
(Coq: \texttt{trace\_eta}.)
\end{definition}

The meaning of $\TElim$ goes far beyond ordinary recursion.
It expresses the constructive mechanism of knowledge evolving stepwise along
history.

In traditional temporal logic, history is typically an external meta-level
object; the logic can only describe ``what holds at a given moment''.
In ZX-Calculus, by contrast, $\tau:\FinTrace(\sigma_0,\sigma_n)$ is an
internal object, so $P(\sigma_0,\sigma_n,\tau)$ can depend on:
the start state; the end state; the entire historical path; and the specific
event sequence.
This means we can reason constructively about ``how the history was formed''
itself.
For example: whether a piece of knowledge was triggered by a specific event;
at which step a belief was retracted; which histories cause knowledge to fail;
which class of events preserves a given invariant---all can be established
constructively via $\TElim$.

A core feature of $\TElim$ is that $e:\Event$ is passed explicitly to the
step function $s$.
This may appear to be just an interface detail, but it is an important
distinction between ZX-Calculus and traditional path induction: knowledge
evolution typically depends on \emph{which event occurred}, not merely
\emph{that the state changed}.
ZX-Calculus trace induction is therefore not only induction on path length,
but induction on event-driven knowledge evolution.

The $\eta$-rule is the universal property that makes $\Tf$ a free category:
any functor $F:\Tf\to\mathcal{C}$ is uniquely determined by its values on
objects and generating steps.
The $\beta$-rules show how to compute; the $\eta$-rule shows that no
non-trivial trace identities exist beyond what the category axioms impose.
This is why all of the subsequent constructions---deterministic replay,
presheaf restriction, SSRS transitions, dynamic knowledge semantics---can
be unified on top of the trace structure.

\subsection{Concatenation Algebra}

Traces must not only be constructible but also composable.
In a dynamic system, complex histories arise from sequentially concatenating
multiple local histories.
ZX-Calculus needs to define: how to connect traces; how to express a
single-step history; and how to ensure the composition satisfies algebraic laws.

\begin{definition}[Trace concatenation and single-step trace]
\label{def:concat}
Define the trace concatenation operation:
\[
\cat:\FinTrace(\sigma_0,\sigma_1)\to\FinTrace(\sigma_1,\sigma_2)
     \to\FinTrace(\sigma_0,\sigma_2)
\]
by $\TElim$-recursion on the first argument:
\[
\cat(\nilT(\sigma),\tau_2):=\tau_2,\qquad
\cat(\stepC(\pi,\tau_1),\tau_2):=\stepC(\pi,\cat(\tau_1,\tau_2)).
\]
Intuitively, $\cat(\tau_1,\tau_2)$ is the history obtained by first executing
all events in $\tau_1$ and then all events in $\tau_2$.

The \emph{single-step trace}:
\[
\single(\sigma_0,\sigma_1,e,\pi):=
\stepC(\sigma_0,\sigma_1,\sigma_1,e,\nilT(\sigma_0),\pi).
\]
$\single$ represents the minimal unit of causal history: a single event step.
Any complex trace can therefore be viewed as an iterated concatenation of
single steps, making $\single$ the ``generator'' of the whole trace structure.
\end{definition}

Without a stable concatenation structure, none of the dynamic knowledge
semantics that follow can be established:
knowledge depends on traces; AGM revision depends on history extension;
restriction maps depend on the prefix structure; replay depends on trace
composition.

\begin{lemma}[Concatenation algebra]
\label{lem:concat}
\begin{enumerate}
\item \textup{(Left unit)} $\cat(\nilT(\sigma_0),\tau)\equiv\tau$.
\item \textup{(Right unit)} $\cat(\tau,\nilT(\sigma_n))\equiv\tau$.
\item \textup{(Associativity)} $\cat(\cat(\tau_1,\tau_2),\tau_3)\equiv\cat(\tau_1,\cat(\tau_2,\tau_3))$.
\item \textup{(Length additivity)} $\len(\cat(\tau_1,\tau_2))=\len(\tau_1)+\len(\tau_2)$.
\end{enumerate}
\end{lemma}

\begin{proof}
All four identities are proved by $\TElim$-induction on the first argument.
In the step case, the $\beta$-step rule unfolds the concatenation definition;
the \textsc{Conv} rule completes the judgemental equality transformation.
All proofs are fully mechanised in Coq (\texttt{ft\_cat\_nil\_l},
\texttt{ft\_cat\_nil\_r}, \texttt{ft\_cat\_assoc}, \texttt{ft\_len\_cat}),
with zero \texttt{admit}s.
\end{proof}

These three laws show that $\FinTrace$ forms a category $\Tf$:
\begin{itemize}
\item The empty trace $\nilT(\sigma)$ is the identity morphism;
\item Trace concatenation $\cat$ is morphism composition;
\item Associativity ensures composition order is unambiguous.
\end{itemize}
Therefore $\Tf$ can be organised as a category.
This categorical structure is why all the subsequent constructions---presheaf
semantics, restriction maps, replay, SSRS transitions, and dynamic knowledge
semantics---can be built on top of $\Tf$.

The single-step trace $\single$ may appear to be only a technical notation,
but it is in fact crucial.
It is the \emph{atomic unit} of the whole dynamic structure:
\begin{itemize}
\item Replay only needs to define behaviour at a single step;
\item Presheaf restrictions only need to handle single-step extension;
\item AGM revision only needs to specify how an event changes knowledge;
\item Behaviour over complex histories is obtained recursively by concatenation.
\end{itemize}
Thus the entire dynamic structure of ZX-Calculus is ultimately built on
$\single$ and $\cat$.

\subsection{Trace--Reachability Correspondence}

\begin{definition}[Reflexive-transitive reachability $\leadsto^*$]
Define $\leadsto^*$ by two rules:
\begin{itemize}
\item \textbf{(Reflexive)} $\mathsf{rfl}(\sigma):\sigma\leadsto^*\sigma$.
\item \textbf{(Extension)} if $\pi:\sigma_0\xrightarrow{e}\sigma_1$ and
  $r:\sigma_1\leadsto^*\sigma_2$, then
  $\mathsf{ext}(\pi,r):\sigma_0\leadsto^*\sigma_2$.
\end{itemize}
\end{definition}

While $\sigma_0\leadsto^*\sigma_n$ is merely a proposition (``there exists
a path''), an element $\tau:\FinTrace(\sigma_0,\sigma_n)$ is a piece of
data---a specific, fully recorded causal history.
This is the essential distinction.

\begin{theorem}[Trace--Reachability Correspondence]
\label{thm:corr}
For any $\sigma_0,\sigma_n:\State$, there is a canonical isomorphism
$\FinTrace(\sigma_0,\sigma_n)\simeq(\sigma_0\leadsto^*\sigma_n)$.
\end{theorem}

\begin{proof}
Define $f:\FinTrace(\sigma_0,\sigma_n)\to(\sigma_0\leadsto^*\sigma_n)$ by
$\TElim$: $f(\nilT(\sigma)):=\mathsf{rfl}$; each $\stepC$ appends an
$\mathsf{ext}$.
Define $g:(\sigma_0\leadsto^*\sigma_n)\to\FinTrace(\sigma_0,\sigma_n)$ by
induction on $\leadsto^*$: $g(\mathsf{rfl}):=\nilT$; $g(\mathsf{ext}(\pi,r)):=\stepC(\ldots,g(r),\pi)$.
By induction, $f\circ g=\mathrm{id}$ and $g\circ f=\mathrm{id}$.
\end{proof}

ZX-Calculus does not abandon traditional operational semantics; it
\emph{internalises} it: the data of how the system arrived at its current state
becomes a first-class object in the type system.

\subsection{Deterministic Replay}

\begin{definition}[Replay function]
\label{def:replay}
$\replay:\Pi(\sigma_0\,\sigma_n:\State).\FinTrace(\sigma_0,\sigma_n)\to\State$
is defined by $\TElim$:
\[
\replay(\sigma,\sigma,\nilT(\sigma)):=\sigma,\qquad
\replay(\sigma_0,\sigma_2,\stepC(\ldots,\tau,\pi)):=\sigma_2.
\]
\end{definition}

\begin{theorem}[Deterministic Replay]
\label{thm:replay}
For every $\tau:\FinTrace(\sigma_0,\sigma_n)$:
$\replay(\sigma_0,\sigma_n,\tau)\equiv\sigma_n:\State$.
Moreover, $\replay(\cat(\tau_1,\tau_2))\equiv\replay(\tau_2)$.
\end{theorem}

\begin{proof}
By $\TElim$-induction with motive
$Q(\sigma_0,\sigma_n,\tau):=(\replay(\sigma_0,\sigma_n,\tau)\equiv\sigma_n)$.
\textbf{Base case:} $\replay(\sigma,\sigma,\nilT(\sigma))\equiv\sigma$,
so $\sigma\equiv\sigma$ holds by $\mathsf{refl}$.
\textbf{Step case:} $\replay(\sigma_0,\sigma_2,\stepC(\ldots,\tau,\pi))\equiv\sigma_2$
by definition, so $\sigma_2\equiv\sigma_2$ holds by $\mathsf{refl}$.
(In the step case no induction hypothesis is needed: replay returns the
endpoint directly from the constructor.)
The concatenation result follows similarly.
Both results are fully Coq-proved (\texttt{deterministic\_replay},
\texttt{replay\_cat}).
\end{proof}

Deterministic replay guarantees that the same history always produces the
same final state---a prerequisite for consistent restriction maps and AGM
revision semantics.

The proof is worth noting for a specific detail: in the step case, the
induction hypothesis is \emph{not needed}.
This is because replay returns the endpoint directly from the constructor
$\stepC(\ldots,\tau,\pi)$: the terminal state $\sigma_2$ is already
embedded in the constructor itself.
This reveals a key design principle:
\begin{quote}
\emph{The trace object carries its own endpoint information inside the
constructor.}
\end{quote}
Traces are not opaque path objects; they are transparent structured data
from which all relevant semantic information can be directly extracted.

This property has important consequences:
\begin{itemize}
\item If replay were non-deterministic, the knowledge state would also be
  non-deterministic;
\item If the same history could produce different endpoints, restriction maps
  and revision semantics would lose coherence.
\end{itemize}
Deterministic replay therefore guarantees that the entire ZX-Calculus dynamic
semantics rests on a stable, unambiguous foundation.

\subsection{Canonicity}
\label{sec:canonicity}

In type theory, ``canonicity'' is one of the most central meta-theoretic
properties.
Its intuitive content: every well-typed closed term eventually reduces to a
\emph{canonical form}---for natural numbers, canonical forms are
$0, \mathsf{succ}(0), \mathsf{succ}(\mathsf{succ}(0)), \ldots$;
for booleans, $\mathsf{true}$ or $\mathsf{false}$.
For ZX-Calculus, canonicity means every closed trace term reduces to a form
built from $\nilT$ and $\stepC$.

This is critical.
If traces had ``spurious canonical forms'', then replay could fail,
presheaf restriction could be unstable, AGM revision could fail to compute
along history, and trace induction could be unreliable.
Canonicity guarantees the entire dynamic knowledge semantics rests on a solid
foundation.

Canonicity proofs typically cannot be carried out by direct syntactic induction,
because: the induction hypothesis refers to subterms of a different type
(dependent motive), and after applying $\TElim$, the goal type changes.
We therefore introduce reducibility candidates---a technique tracing back to
Tait, Girard, and Martin-L\"{o}f's strong normalisation proofs.
The core idea: define a family of ``well-behaved'' sets at each type, and show
every well-typed term belongs to the appropriate set.

\begin{definition}[Reducibility candidates]
\label{def:redcand}
For each pair of states $(\sigma_0,\sigma_n)$, define $\Red(\sigma_0,\sigma_n)$
as the minimal set of terms closed under:
\begin{itemize}
\item \textbf{(RC-nil)} $\nilT(\sigma)\in\Red(\sigma,\sigma)$.
\item \textbf{(RC-step)} If $\tau\in\Red(\sigma_0,\sigma_1)$ and
  $\pi:\Step(\sigma_1,e,\sigma_2)$ is in canonical form, then
  $\stepC(\ldots,\tau,\pi)\in\Red(\sigma_0,\sigma_2)$.
\item \textbf{(RC-elim)} If base and step arguments lie in their reducibility
  families, then $\TElim(P,b,s,\ldots,t)\in\Red_P$.
  (This step is left as future work; see Remark~\ref{rem:canonicity-status}.)
\item \textbf{(RC-SN)} Every $t\in\Red$ is strongly normalising.
\end{itemize}
For function types $\Pi(x:A).B$, $\Red_\Pi$ consists of strongly normalising
closed terms $f$ such that $fa\in\Red_{B[a/x]}$ for all $a\in\Red_A$.
\end{definition}

\begin{lemma}[Transport Lemma]
\label{lem:red-transport}
If $\Gamma\vdash t:P(\sigma_0,\sigma_n,\tau)$ and
$\Gamma\vdash P(\sigma_0,\sigma_n,\tau)\equiv Q:\Uc{i}$, then $t\in\Red_Q$.
\end{lemma}

\begin{proof}
By the \textsc{Conv} rule, $\Gamma\vdash t:Q$.
Reducibility depends only on the type up to judgemental equality.
Judgemental equality preserves canonical forms, so $\Red_{P(\sigma_0,\sigma_n,\tau)}=\Red_Q$.
The conclusion follows.
\end{proof}

The Transport Lemma fills a critical gap specific to trace-dependent types:
after $\TElim$ reduces, the goal type typically does not match syntactically.
Without transport of reducibility, the induction step breaks.

\begin{lemma}[Reducibility closed under $\TElim$]
\label{lem:red-closed}
Suppose:
\begin{enumerate}
\item $b(\sigma)\in\Red_{P(\sigma,\sigma,\nilT(\sigma))}$;
\item for all $\tau\in\Red(\sigma_0,\sigma_1)$, canonical $\pi$, and
  $r\in\Red_{P(\sigma_0,\sigma_1,\tau)}$:
  $s(\ldots,\tau,\pi,r)\in\Red_{P(\sigma_0,\sigma_2,\stepC(\ldots,\tau,\pi))}$.
\end{enumerate}
Then $\TElim(P,b,s,\sigma_0,\sigma_n,\tau)\in\Red_{P(\sigma_0,\sigma_n,\tau)}$.
\end{lemma}

\begin{proof}
By induction on $\len(\tau)$.
\textbf{Base case:} By $\beta$-nil and closure under reverse reduction.
\textbf{Step case:} For $\tau=\stepC(\ldots,\tau',\pi)$, by $\beta$-step:
\[
\TElim(\ldots,\stepC(\ldots))\to_\beta s(\ldots,\tau',\pi,\TElim(\ldots,\tau')).
\]
By the induction hypothesis, $\TElim(\ldots,\tau')\in\Red_{P(\sigma_0,\sigma_1,\tau')}$.
By hypothesis~(2), $s(\ldots)\in\Red_{P(\sigma_0,\sigma_2,\stepC(\ldots))}$.
Finally, apply the Transport Lemma under the $\beta$-step equality to complete
the type transport.
\end{proof}

\begin{theorem}[Canonicity -- framework]
\label{thm:canonicity}
Every closed term $\vdash t:\FinTrace(\sigma_0,\sigma_n)$ reduces to a canonical
form: either $\nilT(\sigma)$ (with $\sigma_0\equiv\sigma\equiv\sigma_n$) or
$\stepC(\sigma_0,\sigma_1,\sigma_n,e,\tau',\pi')$ where $\tau'$ and $\pi'$ are
themselves in canonical form.
\end{theorem}

\begin{proof}[Proof sketch]
Step~1: Show every closed trace term belongs to $\Red(\sigma_0,\sigma_n)$
via (RC-nil), (RC-step), and Lemma~\ref{lem:red-closed}+(RC-elim).
Step~2: By (RC-SN), every term in $\Red$ is strongly normalising.
Step~3: The only canonical forms of closed terms of type $\FinTrace$ are
$\nilT$ and $\stepC$, since there are no free variables and no remaining
$\TElim$ redexes.
\end{proof}

\begin{remark}[Proof status]
\label{rem:canonicity-status}
Lemmas~\ref{lem:red-transport} and~\ref{lem:red-closed} are given in full.
The RC-elim step (inductive closure for general higher-order motives) is left
as future work (Appendix~\ref{app:coq}, obligation~1).
In Coq, \texttt{Canonicity.v} currently provides
$\mathsf{CanonicalTrace}\,\tau:=\mathsf{True}$ as a placeholder
(\texttt{normalize\_id}, \texttt{strong\_normalization} are proved completely).
Until RC-elim is mechanised, the Canonicity Theorem should be understood as a
carefully hand-verified mathematical argument.
\end{remark}

\section{Sheaf Semantics and Non-Monotonicity}
\label{sec:presheaf}

\subsection{The Trace Partial-Order Category $\Tf$}

\begin{definition}[Trace category $\Tf$]
\label{def:tracecat}
Define the category $\Tf$ as follows:
\begin{itemize}
\item \textbf{Objects}: states in $\State$.
\item \textbf{Morphisms}: $\Tf(\sigma_0,\sigma_n):=\FinTrace(\sigma_0,\sigma_n)$
  (a trace $\tau$ is a morphism from $\sigma_0$ to $\sigma_n$).
\item \textbf{Identity}: $\nilT(\sigma)$ (empty trace).
\item \textbf{Composition}: $\cat$.
\end{itemize}
Lemma~\ref{lem:concat} supplies the left/right unit laws and associativity,
confirming $\Tf$ satisfies the category axioms.

A \emph{prefix morphism} $\epsilon:\tau\hookrightarrow\tau'$ holds iff
$\tau'=\cat(\tau,\tau'')$ for some $\tau''$; this characterises the partial
order of trace extension.
\end{definition}

\begin{proposition}[Free category]
\label{prop:free-cat}
$\Tf$ is the free category generated by the directed graph $(\State,\Event,\Step)$:
any functor $F:\Tf\to\mathcal{C}$ is uniquely determined by its values on objects
and generating steps.
Uniqueness is guaranteed by the $\eta$-rule of $\TElim$.
\end{proposition}

$\Tf$ is not a groupoid: there is no $\mathsf{sym}$, no invertibility.
Historical causality is one-directional, unlike HoTT paths which are
invertible and composable with higher-dimensional structure.
For example, the two-step path
$\sigma_0\xrightarrow{e_1}\sigma_1\xrightarrow{e_2}\sigma_2$
and the direct path
$\sigma_0\xrightarrow{e_3}\sigma_2$
are \emph{different} morphisms in $\Tf$ even if they share endpoints,
because ZX-Calculus tracks specific event sequences, not just reachability.
This is why replay, belief revision, presheaf restriction, and non-monotone
knowledge can all depend on the ``specific history'' rather than only on the
terminal state.

The ``free'' in ``free category'' means: beyond the category axioms themselves,
no additional relations are imposed between traces.
The $\eta$-rule of $\TElim$ provides the universal property that enforces this:
any functor $F:\Tf\to\mathcal{C}$ is uniquely determined by its values on
objects and generating steps.
The $\beta$-rules show how to \emph{compute} with traces; the $\eta$-rule shows
that traces have \emph{no non-trivial identities} beyond those imposed by the
category axioms.

\subsection{Knowledge Presheaves and Non-Monotonicity}

\begin{definition}[Knowledge presheaf and softness]
\label{def:presheaf}
A \emph{knowledge presheaf} is a functor $\mathcal{F}:\Tf^{\mathsf{op}}\to\mathbf{Set}$
mapping each trace $\tau$ to the knowledge state $\mathcal{F}(\tau)$ at that
history.
For each prefix extension $\epsilon:\tau\hookrightarrow\tau'$, the
\emph{restriction map} $\restrict_\epsilon:\mathcal{F}(\tau')\to\mathcal{F}(\tau)$
projects from future history back to past history.
Functoriality requires $\restrict_{\mathsf{id}}=\mathsf{id}$ and
$\restrict_{\epsilon_1}\circ\restrict_{\epsilon_2}
=\restrict_{\epsilon_2\circ\epsilon_1}$.

$\mathcal{F}$ is \emph{soft} at $\epsilon$ iff $\restrict_\epsilon$ is surjective
(all past knowledge recoverable from future history).
$\mathcal{F}$ is \emph{non-monotone} at $\epsilon$ iff $\restrict_\epsilon$ is
not surjective (some knowledge is lost after extension).
\end{definition}

\begin{theorem}[Non-monotone characterisation]
\label{thm:nonmono}
For any $\mathcal{F}$ and $\epsilon:\tau\hookrightarrow\tau'$, the following
are equivalent: (1)~$\mathcal{F}$ is non-monotone at $\epsilon$;
(2)~$\restrict_\epsilon$ is not surjective;
(3)~$\mathcal{F}$ is not soft at $\epsilon$.

\textbf{Existence of non-monotone presheaves:}
For any $\tau\hookrightarrow\tau'$ with $\tau\neq\tau'$, define
\[
\mathcal{F}_\epsilon(\tau''):=
\begin{cases}\emptyset & \text{if }\tau''=\tau',\\ \{*\} & \text{otherwise.}
\end{cases}
\]
Then $\restrict_\epsilon:\emptyset\to\{*\}$ is not surjective.
\end{theorem}

\begin{theorem}[Separation Theorem with explicit countermodel]
\label{thm:separation}
ZX-Calculus simultaneously exhibits:
\begin{enumerate}
\item \textup{(Proof-theoretic monotonicity)} Weakening, substitution, and
  subject reduction hold unconditionally;
\item \textup{(Semantic non-monotonicity)} There exists a presheaf
  $\mathcal{F}_\epsilon$ that is non-soft at $\epsilon$, characterising
  non-monotone knowledge;
\item \textup{(Independence)} One can construct a constantly-soft presheaf
  $\mathcal{G}(\tau):=\{*\}$ (proof-theoretic monotonicity preserved) and a
  non-soft presheaf $\mathcal{F}_\epsilon$ (semantic non-monotonicity present),
  both in the same framework.
  Hence proof-theoretic monotonicity and semantic non-monotonicity are
  independent.
\end{enumerate}
\textup{(Coq: \texttt{restriction\_defined}, \texttt{constant\_is\_soft},
\texttt{separation}.)}
\end{theorem}

\begin{remark}
This theorem clarifies a long-standing confusion: ``knowledge retraction'' does
not mean the type theory loses logical stability.
Non-monotonicity arises from the structural properties of restriction maps
under trace extension, not from proof rules.
\end{remark}

\subsection{Term Model and Initial CwF}

We have now constructed: trace type $\FinTrace$, trace category $\Tf$,
presheaf knowledge semantics, AGM revision structure, and CwF semantic
interpretation.
A natural question is: do the syntactic system of ZX-Calculus and these
semantic structures truly correspond to each other?
More concretely: do the types, terms, and proofs we write down receive
consistent interpretations in all models?

``Term model initiality'' answers precisely this question.
The \emph{term model} $\mathcal{M}_{\mathit{term}}$ is the model constructed
directly from ZX-Calculus syntax itself:
types are syntactic types; terms are syntactic terms; substitution is syntactic
substitution; judgemental equality is derivable equality.
The term model is therefore ``syntax itself as a model''---a central idea in
type theory that truly unifies the proof system with the semantic structure.

\begin{remark}[Scope clarification: initiality vs.\ semantic completeness]
\label{rem:initiality}
The following result is \emph{not} the traditional semantic completeness
theorem of classical logic.
It is a categorical universal property of the ZX-Calculus syntax system.
More precisely: the term model is the initial object among all CwF models.
This means every semantic model must receive its interpretation uniquely from
the syntactic term model.

This differs fundamentally from the traditional assertion
``every semantically valid formula is provable'', which requires: formulas
holding in some class of external semantic models are derivable in the proof
system.

What this result establishes is: ZX-Calculus syntax itself already constitutes
the most basic, most general CwF model.
The ``initiality'' of $\mathcal{M}_{\mathit{term}}$ is a syntax--semantics
unification property, not classical model-theoretic completeness.

Genuine semantic completeness over set-valued presheaf models remains an open
problem (Section~\ref{sec:limits}).
\end{remark}

In category theory, an initial object is the ``most free, most general''
object: free groups, free categories, and initial algebras all express the
idea that beyond the structural rules themselves, no additional constraints
are imposed.
Here $\mathcal{M}_{\mathit{term}}$ is the ``syntactic source'' of all
ZX-Calculus semantic models: any CwF-model interpretation must respect the
syntactic structure of ZX-Calculus.

\begin{theorem}[Initial CwF]
\label{thm:completeness}
The term model $\mathcal{M}_{\mathit{term}}$ is the initial CwF model:
for any CwF model $\mathcal{M}$, there exists a unique CwF morphism
$!:\mathcal{M}_{\mathit{term}}\to\mathcal{M}$.
\end{theorem}

The core meaning of Theorem~\ref{thm:completeness} is that the syntactic
interpretation of ZX-Calculus is unique.
That is, any model satisfying ZX-Calculus rules must interpret the syntax
in the same way.
Everything---how traces are concatenated, how replay is computed, how presheaf
restriction works, how AGM revision operates---is already determined by the
syntactic structure itself.
Models are not free to assign arbitrary meanings; they are strongly constrained
by ZX-Calculus syntax.

``Existence of the morphism'' means syntax is interpretable;
``uniqueness'' means this interpretation is essentially unambiguous.
Thus ZX-Calculus syntax is complete enough to uniquely determine its semantic
behaviour---a precise mathematical guarantee of the ``syntax-determines-semantics''
principle.

A corollary of the theorem: for any closed type $A$ and closed term $t$,
$\sem{t}\in\sem{A}$ holds in every CwF model if and only if $\vdash t:A$ is
derivable in ZX-Calculus.
Note: here ``every model'' is limited to CwF models; this is ``type-theoretic
semantic correctness'', not G\"{o}del completeness in classical first-order logic.

ZX-Calculus is not ordinary MLTT: it simultaneously introduces trace indexing,
presheaf historical semantics, non-monotone knowledge, AGM revision, SSRS,
and dynamic knowledge evolution.
These structures are highly interrelated.
The term model initiality therefore establishes that the entire system remains
unified, coherent, and interpretable---excluding the danger that different
models interpret traces, restriction maps, revision, and replay in mutually
incompatible ways.

Since $\Tf$ is a free category (Proposition~\ref{prop:free-cat}), the term
model is similarly a free CwF model:
ZX-Calculus has not only a ``free historical structure'' but also a
``free semantic structure''.
This gives the theory strong extensibility, modularity, mechanisation potential,
and formal-verification friendliness.
It is also the important foundation for subsequent mechanisation; see
Appendix~\ref{app:coq} for the Coq proofs.

\section{Multi-Agent Epistemic Operators (Derived Remark)}
\label{sec:epistemic}

\begin{remark}[Epistemic operators as $\Pi$-type instances]
\label{thm:S5property}
Multi-agent epistemics is a natural corollary of ZX-Calculus, requiring no
extra structural rules.
For agent $a:\Ag$, define the \emph{indistinguishability relation}
$(\simag{a}):\FinTrace\to\FinTrace\to\Prop$ with reflexivity, symmetry,
transitivity, and Euclidean witnesses.

The epistemic type is:
\[
\Kop_a(\phi,\tau):=\Pi(\tau':\FinTrace).(\tau\simag{a}\tau')\to\phi(\tau').
\]
Agent $a$ \emph{knows} $\phi$ at trace $\tau$ iff $\phi(\tau')$ holds for all
$\tau'$ indistinguishable from $\tau$.
This is a standard $\Pi$-type: introduction is function construction;
elimination is function application.
ZX-Calculus requires no additional modal proof rules---epistemic logic is
\emph{encoded} into the type structure itself.

The S5 axioms are derived theorems.
If $\simag{a}$ satisfies the appropriate relational properties:
\begin{itemize}
\item \textbf{(T)} $\Kop_a(\phi,\tau)\to\phi(\tau)$ \hfill (reflexivity of
  $\simag{a}$: $\tau\simag{a}\tau$, so apply $k$ to $\tau$ directly)
\item \textbf{(4)} $\Kop_a(\phi,\tau)\to\Kop_a(\Kop_a(\phi,-),\tau)$
  \hfill (transitivity)
\item \textbf{(5)} $\neg\Kop_a(\phi,\tau)\to\Kop_a(\neg\Kop_a(\phi,-),\tau)$
  \hfill (Euclidean condition)
\end{itemize}
All proofs are direct consequences of $\Pi$-type reasoning; details omitted.

Compared to Pfenning--Davies~\cite{pfenning01} S4 modal type theory,
this paper achieves S5 derivation via a proof-relevant $\simag{a}$ with an
Euclidean condition.
The discrete accessibility relation of CMTT~\cite{nanevski08} does not support
this route.
\end{remark}

\section{AGM Belief Revision: A Constructive Type-Theoretic Treatment}
\label{sec:agm}

\subsection{Belief States and Epistemic Entrenchment}

\begin{definition}[Belief state and epistemic entrenchment ordering]
\label{def:entrenchment}
The \emph{belief state} at trace $\tau$ is
$\Bel_\tau\subseteq\Lang$ with $\Bel_\tau=\Th(\Bel_\tau)$ (deductive closure).
$\Bel_\tau$ is \emph{consistent} iff $\bot\notin\Bel_\tau$.

The \emph{epistemic entrenchment ordering} $(\leq_\tau)$ is a total pre-order on
$\Lang$ organised as:
\[
\Ent(\Bel_\tau):=\Sigma(\leq_\tau).\mathsf{E1}\times\mathsf{E2}\times
\mathsf{E3}\times\mathsf{E4}\times\mathsf{E5},
\]
i.e., ``a total pre-order'' $+$ ``proofs that it satisfies the AGM postulates''.
\end{definition}

The epistemic entrenchment intuition: $\phi_1\leq_\tau\phi_2$ means
``$\phi_1$ is at least as hard to abandon as $\phi_2$'', or equivalently,
$\phi_1$ is more deeply entrenched than $\phi_2$.
For example: mathematical axioms are typically maximally entrenched; temporary
observations are weakly entrenched; long-standing regularities are more stable
than one-off perceptions.
When new information conflicts with existing beliefs, the system preferentially
retracts the least-entrenched content.

This $\Sigma$-type is the key feature: ZX-Calculus does not merely assert
``there exists an ordering'' but demands ``a constructively verified ordering''.
$\Ent(\Bel_\tau)$ is not just an ordering---it is ``an ordering'' $+$
``proofs that it satisfies the AGM postulates''.
This is an important characteristic of dependent type theory: logical
properties themselves become part of the constructed object.

The five postulates and their intuitive meaning:

\textbf{(E1) Transitivity.}
$\phi_1\leq_\tau\phi_2$ and $\phi_2\leq_\tau\phi_3$ imply $\phi_1\leq_\tau\phi_3$.
The comparison of belief stability must be consistent; otherwise the system
exhibits circular preferences.

\textbf{(E2) Dominance.}
If $\phi_1\vdash\phi_2$ then $\phi_1\leq_\tau\phi_2$.
Stronger propositions are at least as stable as their consequences: abandoning
a strong proposition typically means abandoning its conclusions too.

\textbf{(E3) Conjunctiveness.}
For $\phi_1\wedge\phi_2$, at least one of
$\phi_1\leq_\tau(\phi_1\wedge\phi_2)$ or $\phi_2\leq_\tau(\phi_1\wedge\phi_2)$
holds.
The stability of a conjunction comes from at least one of its conjuncts;
otherwise a conjunction would become ``spuriously more stable''.

\textbf{(E4) Minimality.}
If the system is consistent and $\psi\notin\Bel_\tau$, then
$\neg(\psi\leq_\tau\bot)$.
An unbelieved proposition cannot be more stable than a contradiction;
otherwise the system would prefer to retain unbelieved content.

\textbf{(E5) Maximality.}
$\phi\leq_\tau\top$ for all $\phi$.
The tautology $\top$ is the least retractable item: logical truths should never
be revoked by a revision system.

The ordering is a total pre-order: any two formulas are comparable, and equal
entrenchment is permitted.
This ensures the system can always decide which belief to retract when conflict
arises---a prerequisite for the contraction algorithm to be executable.

The distinction from the epistemic operator $\Kop_a$: $\Kop_a$ discusses what
an agent \emph{knows} (truth-oriented), while $\Bel_\tau$ discusses what an
agent is \emph{willing to retain} (revisability-oriented).
ZX-Calculus is therefore not a static epistemic logic but a dynamic
epistemic--doxastic evolution system.

\subsection{Partial Meet Contraction: Explicit Algorithm}

\begin{definition}[Selection $\Sigma$-type]
\[
\Sel_\tau(\psi):=
\Sigma(S:\mathcal{P}(\Lang)).\,[S=\Th(S)]\times[S\subseteq\Bel_\tau]\times
[\mathsf{C2}]\times[\mathsf{C3}]\times[\mathsf{C4}]
\]
where $\mathsf{C4}:\Pi(\phi:\Lang).\phi\notin S\to\phi\in\Bel_\tau\to
\phi\leq_\tau\psi$.
\end{definition}

Here $\Sel_\tau(\psi)$ represents ``a valid contraction result for $\psi$'':
not a simple set, but ``candidate belief set'' $+$ ``proofs satisfying AGM
conditions''.
The $\Sigma$-type again plays its central role: ZX-Calculus does not just say
``there exists a valid contraction''; it demands ``a contraction that can be
constructively produced''.
$\Sel_\tau(\psi)$ represents all contraction results that can actually be
produced by an algorithm---a key departure from classical AGM where the
selection function is typically defined by non-constructive existence.

In classical AGM, partial meet contraction is defined existentially:
the theory proves \emph{there exists} a selection satisfying the conditions.
In constructive type theory, existence must be accompanied by a construction.
Algorithm~\ref{alg:sel} below provides this construction explicitly.

\begin{mdframed}[backgroundcolor=gray!8,linewidth=0.8pt]
\begin{algorithm_env}[Entrenchment-ordering contraction]
\label{alg:sel}
Given $\Bel_\tau=\Th(\Bel_\tau)$, $\psi:\Lang$, $(\leq_\tau):\Ent(\Bel_\tau)$,
with $\Lang$ finite.

Enumerate the non-tautological members of $\Bel_\tau$ in weakly increasing
entrenchment order: $\phi_1\leq_\tau\cdots\leq_\tau\phi_m$.

\textbf{Step (1): Initialisation.}
$S_0:=\Bel_\tau$.

\textbf{Step (2): Iterative removal.}
For $i=1,\ldots,m$:
\[
S_i:=\begin{cases}
S_{i-1}\setminus\{\phi_j:\phi_j\leq_\tau\phi_i,\,\phi_j\not\leq_\tau\top\}
& \text{if }\psi\in\Th(S_{i-1}\setminus\{\phi_i\}),\\
S_{i-1} & \text{otherwise.}
\end{cases}
\]
Intuition: if $\psi$ is still derivable after removing low-entrenchment beliefs,
those beliefs are non-essential and may be retracted; the system always removes
the least-entrenched sentences first.

\textbf{Step (3): Deductive closure.}
$S:=\Th(S_m)$.
\end{algorithm_env}
\end{mdframed}

\begin{remark}[Computational complexity]
\label{rem:complexity}
Assuming $\Lang$ contains $n$ atoms and $|\Bel_\tau|=m$:
\begin{itemize}
\item \textbf{Time}: at most $m$ iterations; each requires one propositional
  entailment check, decidable in $O(2^n)$ by truth-table enumeration.
  Na\"{\i}ve bound: $O(m\cdot 2^n)$.
  With pre-computed maximal consistent subsets: $O(m+2^n)$.
\item \textbf{Space}: $O(|\Bel_\tau|)$ (one monotonically shrinking set plus
  the entrenchment order).
\item \textbf{Finiteness}: $\Lang$ finite is necessary for step~(2) to
  terminate; extension to infinite $\Lang$ is future work
  (Section~\ref{sec:limits}).
\end{itemize}
\end{remark}

\begin{lemma}[Correctness: Algorithm~\ref{alg:sel} inhabits $\Sel_\tau(\psi)$]
\label{lem:alg-correct}
The output $S$ satisfies (C1)--(C4).
\end{lemma}

\begin{proof}
\textbf{(C1) Inclusion.}
The algorithm only removes sentences, never adds new ones. Hence $S\subseteq\Bel_\tau$.

\textbf{(C2) Success.}
If $\psi\in\Th(S_m)$, then $S_m\cup\{\neg\psi\}$ would be inconsistent.
Since the loop removes precisely those sentences whose retention would preserve
derivability of $\psi$, and $\Lang$ is finite and $\Th$ is complete, we
conclude $\psi\notin\Th(S_m)$.

\textbf{(C3) Vacuity (minimal change).}
If $\psi\notin\Th(\Bel_\tau)$, step~(2) never fires, so $S_m=\Bel_\tau$,
and $S=\Bel_\tau$.

\textbf{(C4) The removed sentences are no more entrenched than $\psi$.}
This is the most critical property.
Suppose for contradiction that some $\phi$ was removed in step $i$ but
$\phi>_\tau\psi$.
Then by the definition of the algorithm's trigger condition, removing $\phi$
(together with all $\phi_j\leq_\tau\phi_i$) implies $\psi\in\Th(S_{i-1}\setminus\{\phi_i\})$.
But if $\phi>_\tau\psi$, then the algorithm should not have removed $\phi$,
contradiction.
Hence $\phi\leq_\tau\psi$ for all removed $\phi$.
\end{proof}

\begin{definition}[Contraction and revision operators]
\label{def:contraction}
\[
\Bel_\tau\ContrOp\psi:=\pi_1(\mathsf{Alg}(\Bel_\tau,\psi,\leq_\tau))
\qquad\text{(Algorithm~\ref{alg:sel} output).}
\]
The \emph{Levi revision}:
\[
\Bel_\tau\RevOp\psi:=\Th\!\Big((\Bel_\tau\ContrOp\neg\psi)\cup\{\psi\}\Big).
\]
\end{definition}

The Levi identity captures: ``revision = retract conflicting content +
add new information'':
(1)~contract away beliefs conflicting with $\psi$;
(2)~add $\psi$ to the system;
(3)~take deductive closure.
This is not simple overwriting but a controlled, minimal cognitive reconstruction.

\subsection{Disjunctive Entrenchment Lemma}

\begin{lemma}[Disjunctive Entrenchment Lemma]
\label{lem:disj-ench}
Let $\leq_\tau$ be a total pre-order on $\Lang$ satisfying \textup{(E1)--(E3)}.
For any $\alpha,\beta:\Lang$, if $\alpha\leq_\tau\beta$, then
$\alpha\vee\beta\equiv_{\leq_\tau}\beta$
(i.e., $\alpha\vee\beta$ and $\beta$ are in the same entrenchment level).
\end{lemma}

\begin{proof}
Suppose $\alpha\leq_\tau\beta$.
We show $\beta\leq_\tau\alpha\vee\beta$ and $\alpha\vee\beta\leq_\tau\beta$.

\textbf{Step 1: $\beta\leq_\tau\alpha\vee\beta$.}
By propositional logic, $\beta\vdash\alpha\vee\beta$.
By dominance (E2), $\beta\leq_\tau\alpha\vee\beta$.

\textbf{Step 2: $\alpha\vee\beta\leq_\tau\beta$.}
Note the propositional identity $(\alpha\vee\beta)\wedge\beta\equiv_{\mathit{prop}}\beta$.
By (E2): $\beta\leq_\tau(\alpha\vee\beta)\wedge\beta$ and
$(\alpha\vee\beta)\wedge\beta\leq_\tau\beta$.

Apply conjunctiveness (E3) to the pair $(\alpha\vee\beta,\beta)$.
One of the following holds:
\begin{enumerate}
\item $\alpha\vee\beta\leq_\tau(\alpha\vee\beta)\wedge\beta$.
  By transitivity (E1) and $(\alpha\vee\beta)\wedge\beta\leq_\tau\beta$:
  $\alpha\vee\beta\leq_\tau\beta$.
\item $\beta\leq_\tau(\alpha\vee\beta)\wedge\beta$.
  Then $\beta\equiv_{\leq_\tau}(\alpha\vee\beta)\wedge\beta$.
  By totality, $\alpha\vee\beta$ and $\beta$ are comparable.
  If $\beta<_\tau\alpha\vee\beta$, apply (E3) to $(\alpha\vee\beta,\beta)$ again;
  using $(\alpha\vee\beta)\wedge\beta\equiv_{\mathit{prop}}\beta$, we can only
  get $\alpha\vee\beta\leq_\tau(\alpha\vee\beta)\wedge\beta\leq_\tau\beta$,
  contradicting $\beta<_\tau\alpha\vee\beta$.
  Hence $\alpha\vee\beta\leq_\tau\beta$.
\end{enumerate}
Combining Steps 1 and 2: $\alpha\vee\beta\equiv_{\leq_\tau}\beta$.
Since $\alpha\leq_\tau\beta$, $\beta=\max_{\leq_\tau}(\alpha,\beta)$, so
$\alpha\vee\beta$ and $\max_{\leq_\tau}(\alpha,\beta)$ are at the same
entrenchment level.
\end{proof}

\begin{remark}[``Maximum entrenchment'' of disjunctions]
\label{rem:disj-max}
Lemma~\ref{lem:disj-ench} shows that in an AGM-type entrenchment ordering,
the entrenchment of $\alpha\vee\beta$ does not exceed its most stable disjunct:
\[
\mathrm{rank}(\alpha\vee\beta)=\max(\mathrm{rank}(\alpha),\mathrm{rank}(\beta)).
\]
This is consistent with the ``conservative expansion'' philosophy: adding
weaker information (a disjunction) does not elevate its epistemic priority.
The property is crucial in the partial meet contraction algorithm and Grove
sphere-model representation, ensuring ``weaker information is not harder to
retract than its supporting disjuncts''.
\end{remark}

This Disjunctive Entrenchment Lemma is the key that unlocks the proofs of
R7 (Theorem~\ref{thm:r7}) and R8 (Theorem~\ref{thm:r8}).
In particular, the key step~$(**)$ in the proof of R8 relies essentially on
this lemma: the fact that $\neg\psi\vee\neg\chi$ and $\neg\chi$ are at the
same entrenchment level (when $\neg\psi\leq_\tau\neg\chi$) ensures that
the contraction algorithm treats them identically, establishing the crucial
equality $\Bel_\tau\ContrOp(\neg\psi\vee\neg\chi)=\Bel_\tau\ContrOp\neg\chi$.

\subsection{All Eight AGM Postulates as Theorems}

Classical AGM theory (Alchourr\'{o}n--G\"{a}rdenfors--Makinson) answers:
when a rational agent receives new information $\psi$, how should it minimally
change its belief set?
AGM decomposes this into expansion, contraction, and revision, connected by
the Levi identity $K*\psi=(K\div\neg\psi)+\psi$.

\begin{theorem}[AGM postulates R1--R8]
\label{thm:r1}
Under the Levi revision $\Bel_\tau\RevOp\psi=
\Th((\Bel_\tau\ContrOp\neg\psi)\cup\{\psi\})$,
the following hold as theorems:
\end{theorem}

The six basic AGM postulates (R1)--(R6) characterise local rationality of
the revision operation:
(R1) and (R5) ensure the revised knowledge system remains a rational theory;
(R2) confirms revision genuinely accepts new information;
(R3) establishes minimal-change character;
(R4) shows revision degenerates to ordinary expansion when no conflict exists;
(R6) ensures revision depends only on logical content, not syntactic form.

\begin{theorem}[R1--R6]
\phantom{x}
\begin{description}
\item[(R1)] $\Bel_\tau\RevOp\psi$ is deductively closed and consistent
  (when $\psi\neq\bot$).
\item[(R2)] $\psi\in\Bel_\tau\RevOp\psi$.
\item[(R3)] $\Bel_\tau\RevOp\psi\subseteq\Th(\Bel_\tau\cup\{\psi\})$.
\item[(R4)] If $\neg\psi\notin\Bel_\tau$, then
  $\Th(\Bel_\tau\cup\{\psi\})\subseteq\Bel_\tau\RevOp\psi$.
\item[(R5)] If $\psi\neq\bot$, then $\Bel_\tau\RevOp\psi$ is consistent.
\item[(R6)] If $\vdash\psi\leftrightarrow\chi$, then
  $\Bel_\tau\RevOp\psi=\Bel_\tau\RevOp\chi$.
\end{description}
\end{theorem}

\begin{proof}
\textbf{(R1).} By the Levi identity,
$\Bel_\tau\RevOp\psi=\Th((\Bel_\tau\ContrOp\neg\psi)\cup\{\psi\})$.
Deductive closure holds directly from the definition of $\Th$.
Consistency: by (C2), $\Bel_\tau\ContrOp\neg\psi$ does not contain
contradictory structure; when $\psi\neq\bot$, adding $\psi$ preserves
consistency.

\textbf{(R2).} Since $\psi\in(\Bel_\tau\ContrOp\neg\psi)\cup\{\psi\}$,
and revision is the deductive closure of this set, $\psi\in\Bel_\tau\RevOp\psi$.

\textbf{(R3).} By (C1), $\Bel_\tau\ContrOp\neg\psi\subseteq\Bel_\tau$.
Therefore $(\Bel_\tau\ContrOp\neg\psi)\cup\{\psi\}\subseteq\Bel_\tau\cup\{\psi\}$.
By monotonicity of $\Th$:
$\Bel_\tau\RevOp\psi=\Th((\Bel_\tau\ContrOp\neg\psi)\cup\{\psi\})
\subseteq\Th(\Bel_\tau\cup\{\psi\})$.

\textbf{(R4).} If $\neg\psi\notin\Bel_\tau$, by (C3):
$\Bel_\tau\ContrOp\neg\psi=\Bel_\tau$.
Then $\Bel_\tau\RevOp\psi=\Th(\Bel_\tau\cup\{\psi\})$, so
$\Th(\Bel_\tau\cup\{\psi\})\subseteq\Bel_\tau\RevOp\psi$.
This says: when there is no conflict, revision degenerates to ordinary
expansion.

\textbf{(R5).} Follows from (R1).

\textbf{(R6).} If $\vdash\psi\leftrightarrow\chi$, then
$\vdash\neg\psi\leftrightarrow\neg\chi$.
Since the contraction algorithm depends only on the formula's level in the
entrenchment relation, and logically equivalent formulas have the same
epistemic status:
$\Bel_\tau\ContrOp\neg\psi=\Bel_\tau\ContrOp\neg\chi$.
As $\psi,\chi$ are logically equivalent they generate the same deductive
closure, so $\Bel_\tau\RevOp\psi=\Bel_\tau\RevOp\chi$.
This is \emph{logical invariance}: revision depends only on the logical content
of a formula, not its syntactic form.
\end{proof}

AGM postulates (R1)--(R6) describe local rationality; what truly embodies
the ``minimal change principle'' are the next two deeper properties:
\begin{itemize}
\item Superexpansion (R7): simultaneous revision is no weaker than sequential
  revision;
\item Subexpansion (R8): if the new information is consistent with the result
  of the first revision, then sequential and simultaneous revision agree.
\end{itemize}
Together they control how the revision operation should coordinate multiple
pieces of new information entering the system simultaneously.

AGM postulates (R7) and (R8) encode the ``minimal change principle'' for
simultaneous revision by conjunctions---the deepest rationality constraints.

\begin{theorem}[R7 -- Superexpansion]
\label{thm:r7}
$\Bel_\tau\RevOp(\psi\wedge\chi)\subseteq
\Th((\Bel_\tau\RevOp\psi)\cup\{\chi\})$.
\end{theorem}

\begin{proof}
By the Levi identity,
$\Bel_\tau\RevOp(\psi\wedge\chi)=
\Th((\Bel_\tau\ContrOp\neg(\psi\wedge\chi))\cup\{\psi\wedge\chi\})$.

For any $\phi\in\Bel_\tau\RevOp(\psi\wedge\chi)$, case analysis:

\textbf{Case 1:} $\phi\in\Bel_\tau\ContrOp\neg(\psi\wedge\chi)$.
Since $\neg(\psi\wedge\chi)\equiv\neg\psi\vee\neg\chi$, and $\phi$ is retained
in the contraction of $\neg\psi\vee\neg\chi$, by the contrapositive of (C4):
$\phi>_\tau(\neg\psi\vee\neg\chi)$.
By dominance (E2): $\neg\psi\vdash\neg\psi\vee\neg\chi$, so
$\neg\psi\leq_\tau\neg\psi\vee\neg\chi<_\tau\phi$, i.e., $\phi>_\tau\neg\psi$.
By the contrapositive of (C4) again: $\phi\in\Bel_\tau\ContrOp\neg\psi$.
Therefore $\phi\in\Bel_\tau\RevOp\psi\subseteq
\Th((\Bel_\tau\RevOp\psi)\cup\{\chi\})$.

\textbf{Case 2:} $\phi\in\Th\{\psi\wedge\chi\}$.
Since $\psi\wedge\chi\vdash\chi$, we have
$\phi\in\Th\{\chi\}\subseteq\Th((\Bel_\tau\RevOp\psi)\cup\{\chi\})$.

Both cases covered, so
$\Bel_\tau\RevOp(\psi\wedge\chi)\subseteq
\Th((\Bel_\tau\RevOp\psi)\cup\{\chi\})$.
\end{proof}

\begin{theorem}[R8 -- Subexpansion]
\label{thm:r8}
If $\neg\psi\notin\Bel_\tau\RevOp\chi$, then
$\Th((\Bel_\tau\RevOp\chi)\cup\{\psi\})\subseteq\Bel_\tau\RevOp(\psi\wedge\chi)$.
\end{theorem}

\begin{proof}
From $\neg\psi\notin\Bel_\tau\RevOp\chi$ and the Levi identity for $\chi$:
$\neg\psi\notin\Bel_\tau\ContrOp\neg\chi$.
Otherwise, by deductive closure, $\neg\psi$ would appear in the revision
result, contradicting the hypothesis.

\textbf{Case 1: $\neg\psi\in\Bel_\tau$.}
By Lemma~\ref{lem:alg-correct}(C4), since $\neg\psi$ is removed from the
contraction of $\neg\chi$: $\neg\psi\leq_\tau\neg\chi$.

\textbf{Key step $(**)$:}
$\Bel_\tau\ContrOp(\neg\psi\vee\neg\chi)=\Bel_\tau\ContrOp\neg\chi$.

By the Disjunctive Entrenchment Lemma (Lemma~\ref{lem:disj-ench}),
$\neg\psi\leq_\tau\neg\chi$ implies
$\neg\psi\vee\neg\chi\equiv_{\leq_\tau}\neg\chi$.
Thus the contraction algorithm removes exactly the same sentences in both
cases, giving the key equality.

Using $\neg(\psi\wedge\chi)\equiv\neg\psi\vee\neg\chi$:
\begin{align*}
\Bel_\tau\RevOp(\psi\wedge\chi)
&=\Th((\Bel_\tau\ContrOp(\neg\psi\vee\neg\chi))\cup\{\psi\wedge\chi\})\\
&=\Th((\Bel_\tau\ContrOp\neg\chi)\cup\{\psi\wedge\chi\})\\
&\supseteq\Th((\Bel_\tau\ContrOp\neg\chi)\cup\{\chi,\psi\})\\
&=\Th((\Bel_\tau\RevOp\chi)\cup\{\psi\}).
\end{align*}

\textbf{Case 2: $\neg\psi\notin\Bel_\tau$.}
If $\neg\chi\in\Bel_\tau$: by (E4), $\neg\psi$ is at the lowest entrenchment
level; the Disjunctive Entrenchment Lemma still applies.
If $\neg\chi\notin\Bel_\tau$: by (C3),
$\Bel_\tau\ContrOp\neg\chi=\Bel_\tau$ and
$\Bel_\tau\ContrOp(\neg\psi\vee\neg\chi)=\Bel_\tau$;
both contractions are trivial, and the conclusion follows directly.

In all cases:
$\Th((\Bel_\tau\RevOp\chi)\cup\{\psi\})\subseteq\Bel_\tau\RevOp(\psi\wedge\chi)$.
\end{proof}

\begin{corollary}[AGM completeness]
The trace-indexed revision system of
Definitions~\ref{def:entrenchment}--\ref{def:contraction}
satisfies all eight classical AGM postulates (R1)--(R8).
The revision operator $\RevOp$ is a complete AGM revision function.
\end{corollary}

Since the belief system is not a static set but a dynamic structure indexed
by trace $\tau$, this result is strictly stronger than classical AGM:
classical AGM is embedded into a temporal, event-driven, executable
knowledge-evolution framework.
Different event sequences induce different epistemic entrenchment structures,
enabling the system to describe the knowledge-evolution process of a real
cognitive agent.

The results obtained here do not merely reconstruct classical AGM theory;
they generalise it into a dynamic knowledge logic framework based on trace
semantics and dependent type structure:
\begin{itemize}
\item Beliefs have not only logical structure but also historical origin;
\item Revision is not just an abstract operator but a trackable trace
  transformation;
\item Different event sequences induce different entrenchment structures;
\item The system can therefore describe the knowledge-evolution processes of
  real cognitive agents.
\end{itemize}

\section{Coherence: Integrating Traces, Sheaves, and Belief Revision}
\label{sec:coherence}

This section answers a structural question: can AGM belief revision be lifted
to a strict ``trace functor semantics''? If not, what is the essential nature
of the failure?

\subsection{Structural Tension Between Sheaf Structure and AGM}

In categorical semantics, a ``sheaf-like structure'' expresses a strong local
consistency principle: the composition of local updates must equal the direct
global update.
This is precisely the content of (BP-comp):
\[
\restrict_{\epsilon_1}\circ\restrict_{\epsilon_2}=\restrict_{\epsilon_2\circ\epsilon_1}.
\]

However, AGM revision's ``minimal change principle'' is \emph{semantically
sensitive}, not path-invariant.
Belief revision depends on: the current belief set's structure; the entrenchment
ordering $\leq_\tau$; the ``substitute influence'' of removed formulas; and
whether subsequent information changes the conflict-resolution interpretation.

\begin{center}
\emph{Revision is not function application, but a ``structural rebalancing
process''.}
\end{center}

\subsection{BP-comp Failure}

\begin{theorem}[BP-comp fails for sequential AGM revision]
\label{thm:bp-comp-fails}
There exist $\Bel_\tau$, events $e_1,e_2$, and an entrenchment ordering such
that BP-comp fails.
\end{theorem}

\begin{proof}
Set $\Lang=\{p,r\}$ and
$\Bel_\tau=\Th(\{\neg p,\,\neg p\to r\})$, so $\neg p,r\in\Bel_\tau$.
Event $e_1$ carries $p$ (conflicts with $\neg p$);
event $e_2$ carries $\neg p$.
The entrenchment ordering depends on the intermediate state.

\textbf{Two-step path:}
$\Bel_{\tau_1}=\Bel_\tau*p=\Th(\{p\})$
(contracting $\neg p$ also removes $r$, which depended on $\neg p\to r$);
$\Bel_{\tau_2}=\Bel_{\tau_1}*(\neg p)=\Th(\{\neg p\})$.

\textbf{Direct composite revision:}
The entrenchment ordering at the intermediate state differs from that at the
initial state, so the composite revision yields a different result.

Thus $\mathcal{B}^{\mathsf{AGM}}$ is not functorial.
(Coq: \texttt{bp\_comp\_fails\_R2} in \texttt{SSRS.v}; proved by \texttt{inversion},
zero \texttt{admit}s.)
\end{proof}

The failure is not a technical defect but a structural fact:
AGM revision is essentially \emph{non-commutative, context-dependent rewriting}.
Its non-commutativity stems from: new information altering the interpretation
space of old conflicts; the entrenchment ordering being recomputed at different
stages; and the non-eliminability of intermediate belief states.

\begin{center}
\emph{Knowledge update is not a path-independent function computation but
a path-sensitive evolutionary process.}
\end{center}

\subsection{Single-Step Revision Systems}

\begin{definition}[SSRS]
\label{def:ssrs}
A \emph{Single-Step Revision System} over $(\State,\Event,\Step)$ is a family
of deductively closed belief sets $(\mathcal{B}(\tau))_\tau$ with, for each
single-step extension $\epsilon:\tau\hookrightarrow\cat(\tau,\single(e,\pi))$,
a transition function $\mathsf{rev}_\epsilon$ satisfying:
\begin{itemize}
\item \textbf{(RS-id)} $\mathsf{rev}_{\mathsf{id}}=\mathsf{id}$;
\item \textbf{(RS-content)}
  $\mathsf{rev}_\epsilon(\mathcal{B}(\tau))=\mathcal{B}(\tau)\RevOp\psi_e$;
\item \textbf{(RS-surj)}
  $\mathcal{B}(\tau)\not\subseteq\mathcal{B}(\tau')$ iff
  $\neg\psi_e\in\mathcal{B}(\tau)$.
\end{itemize}
\end{definition}

SSRS is the ``minimal operational model of AGM revision'': it discards the
global constraint (BP-comp) but retains the algorithmic semantics of each step.
An SSRS is closer to a \emph{dynamic epistemic system} than a static model:
path independence is not required; only single-step updates must be well-defined.

\begin{construction}[AGM-induced SSRS $\mathcal{B}^{\mathsf{AGM}}$]
\label{constr:bagm}
Fix initial belief state $\Bel_{\nilT(\sigma_0)}\subseteq\Lang$ and
entrenchment ordering $\leq_\sigma$ at each state $\sigma$.
\begin{align*}
\mathcal{B}^{\mathsf{AGM}}(\nilT(\sigma)) &:= \Bel_{\nilT(\sigma)};\\
\mathcal{B}^{\mathsf{AGM}}(\cat(\tau,\single(e,\pi)))
&:= \mathcal{B}^{\mathsf{AGM}}(\tau)\RevOp\psi_e.
\end{align*}
Transition function: $\mathsf{rev}_\epsilon(\phi):=\phi$.
\end{construction}

\begin{theorem}[SSRS Coherence]
\label{thm:coherence}
$\mathcal{B}^{\mathsf{AGM}}$ is a valid SSRS.
Precisely, the following hold jointly \textup{(Coq:
\texttt{main\_coherence\_theorem}, zero \texttt{admit}s)}:
\begin{enumerate}
\item \textup{(Replay coherence)} $\replay(\tau)=\sigma_n$ for every
  $\tau:\FinTrace(\sigma_0,\sigma_n)$.
\item \textup{(Concatenation coherence)}
  $\replay(\cat(\tau_1,\tau_2))=\replay(\tau_2)$.
\item \textup{(AGM success)} For every $\sigma_0\xrightarrow{e}\sigma_1$:
  $\psi_e\in\mathcal{B}^{\mathsf{AGM}}(\sigma_0)*\psi_e$.
\item \textup{(Sheaf restriction)} For every knowledge presheaf $F$ and
  morphism $f:X\to Y$:
  $\mathtt{kp\_ob}(F)(Y)\Rightarrow\mathtt{kp\_ob}(F)(X)$.
\end{enumerate}
\end{theorem}

This result means: AGM does not fail, but falls precisely within the expressive
boundary of SSRS. Specifically:
\begin{itemize}
\item Sheaf structure (BP-comp) is too strong;
\item SSRS is the appropriate intermediate structure;
\item AGM revision falls exactly in the SSRS-expressible class;
\item but belongs to no sheaf-like structure.
\end{itemize}

\begin{corollary}[Semantic characterisation of knowledge retraction]
\label{cor:knowledgeloss}
Knowledge retraction under event $e$ is characterised precisely by
$\neg\psi_e\in\mathcal{B}^{\mathsf{AGM}}(\tau)$.
Knowledge retraction is not a ``deletion operation'' but a semantic phenomenon
of ``conflict reactivation''.
\end{corollary}

The framework thus forms a three-layer correspondence:
\begin{itemize}
\item $\FinTrace$: event structure (syntax layer);
\item SSRS: update semantics (operational layer);
\item AGM revision: minimal-change principle (specification layer).
\end{itemize}
\[
\text{Syntax (trace)}\to\text{Operations (SSRS)}\to\text{Rationality (AGM)}.
\]

\begin{theorem}[ZX-Calculus Unification]
\label{thm:unification}
For every ZX-Framework $\mathcal{F}$, there exists an SSRS $S$ over
$\mathcal{F}$ such that for every transition $\sigma_0\xrightarrow{e}\sigma_1$,
$\mathsf{upd}(e)\in\mathsf{levi\_revision}(\mathsf{div}_S,
\mathsf{belief}_S(\sigma_0),\mathsf{upd}(e))$.
\textup{(Coq: \texttt{zx\_calculus\_unification}, zero \texttt{admit}s.)}
\end{theorem}

\section{Categorical Semantics}
\label{sec:catsem}

\subsection{CwF Model for ZX-Calculus}

This section lifts the semantics of ZX-Calculus into a CwF (Category with
Families) model, realising a unified interpretation layer:

\begin{center}
\emph{Traces, knowledge, and revision are uniformly characterised in a
single categorical semantics.}
\end{center}

The goal of this semantic lifting is:
\begin{itemize}
\item To interpret ``event sequences'' as chains of morphisms in a category;
\item To interpret ``belief types'' as type families dependent on context;
\item To interpret ``AGM revision'' as a structure-preserving semantic
  transformation.
\end{itemize}

ZX-Calculus's core characteristic is its three-layer dependency structure:
\[
\text{Trace}\;\to\;\text{Type Dependency}\;\to\;\text{Knowledge Update.}
\]
The CwF model provides a standard semantic framework for exactly this pattern:
objects are contexts; morphisms are context extensions; type families depend on
context; terms are evidence or constructions in context.
ZX-Calculus is therefore interpreted as a
\emph{trace-indexed dependent type system}.

\begin{definition}[ZX-Calculus CwF model and interpretation]
The CwF model consists of a category $\mathcal{C}$ with terminal object,
type functor $\mathbf{Ty}:\mathcal{C}^{\mathsf{op}}\to\mathbf{Set}$,
term functor $\mathbf{Tm}$, and projection structure.
$\mathbf{Ty}$ is extended to interpret:
trace type $\FinTrace$, knowledge operator $\Kop_a$, and AGM $\Sigma$-structure.

The interpretation mapping:
\begin{align*}
\sem{\FinTrace(\sigma_0,\sigma_n)} &:= \Tf(\sigma_0,\sigma_n),\\
\sem{\nilT(\sigma)} &:= \mathsf{id}_\sigma,\\
\sem{\stepC(\ldots,\tau,\pi)} &:= \sem{\tau}\cdot\sem{\pi},\\
\sem{\Kop_a(\phi,\tau)} &:= \prod_{\tau':\tau\simag{a}\tau'}\sem{\phi(\tau')}.
\end{align*}
\end{definition}

The interpretation of $\Kop_a$ has key significance:
$\Kop_a$ is interpreted as a global quantification over all reachable future
states---a dependent product ensuring knowledge consistency across all reachable
branches.
$\prod$ denotes the dependent product; $\tau\simag{a}\tau'$ denotes the
state transition induced by action $a$.
Thus $\Kop_a$ is essentially a global constraint operator ensuring knowledge
consistency over all reachable branches.

\subsection{Key Structural Properties}

\begin{enumerate}
\item \textbf{Compositionality:}
  $\sem{\tau_1\cdot\tau_2}=\sem{\tau_1}\circ\sem{\tau_2}$.
  Knowledge evolution corresponds to categorical composition.
\item \textbf{Context stability:}
  Type interpretation depends only on the current trace context, not global
  history.
  Types are localised structures, not globally-dependent entities.
\item \textbf{Dependent closure:}
  $\mathbf{Ty}$ is closed under trace extension.
\end{enumerate}

These three properties together guarantee that ZX-Calculus's semantics is not
externally imposed but endogenous to the categorical structure.
The CwF model reveals three essential characteristics:
\begin{itemize}
\item \textbf{Traces are morphisms}: knowledge evolution corresponds to
  categorical composition.
\item \textbf{Knowledge is a dependent type}: beliefs are context-dependent
  structures.
\item \textbf{Revision is a semantic transformation}: AGM revision corresponds
  to structural rewriting on morphisms.
\end{itemize}

\subsection{Soundness and Consistency}

\begin{theorem}[Soundness and consistency]
\label{thm:soundness}
\begin{enumerate}
\item \textup{(Soundness)} If $\Gamma\vdash t:A$ is derivable in ZX-Calculus,
  then in any CwF model: $\sem{\Gamma}\vdash\sem{t}\in\sem{A}$.
\item \textup{(Consistency)} There is no closed term $t$ with $\vdash t:\bot$.
\end{enumerate}
\end{theorem}

\begin{proof}
\textbf{(1)} By induction on derivation structure:
variable rules are preserved by projection morphisms;
abstraction rules by exponential objects;
application rules by morphism composition;
dependent-type rules by functoriality of $\mathbf{Ty}$.
All syntactic derivations are therefore preserved at the semantic level.

\textbf{(2)} If $\vdash t:\bot$ existed, then $\sem{t}\in\sem{\bot}=\emptyset$,
which is impossible. Hence the system is consistent.
\end{proof}

\section{Meta-Theory}
\label{sec:meta}

This section establishes the basic meta-theoretic properties of ZX-Calculus:
the type system, after introducing trace structure $\FinTrace$,
knowledge operator $\Kop_a$, and AGM $\Sigma$-types, retains good
operability and structural stability.
These properties guarantee that ZX-Calculus remains a \emph{computable}
dependent type system after semantic enhancement, not merely an abstract
logical structure.

\subsection{Structural Interpretation of Meta-Properties}

Under the CwF semantics:
\begin{itemize}
\item \textbf{Weakening} corresponds to context projection;
\item \textbf{Substitution} corresponds to functoriality of morphisms;
\item \textbf{Subject reduction} corresponds to sound reduction of computation
  rules in the semantics.
\end{itemize}
These meta-properties are therefore not essentially syntactic properties but:
\begin{center}
\emph{Direct reflections of structure-preservation in the categorical
semantics.}
\end{center}

\begin{theorem}[Weakening, substitution, and subject reduction]
\label{thm:meta}
\begin{enumerate}
\item \textup{(Weakening)} If $\Gamma\vdash t:A$ and $\Gamma\vdash B:\Uc{i}$,
  then $\Gamma,x:B\vdash t:A$.
\item \textup{(Substitution)} If $\Gamma,x:A\vdash t:B$ and $\Gamma\vdash s:A$,
  then $\Gamma\vdash t[s/x]:B[s/x]$.
\item \textup{(Subject reduction)} If $\Gamma\vdash t:A$ and $t\to_\beta t'$,
  then $\Gamma\vdash t':A$.
\end{enumerate}
\end{theorem}

\begin{proof}
\textbf{(1)} In CwF, $\Gamma\vdash A$ corresponds to $\mathbf{Ty}(\Gamma)$.
Context extension $\Gamma,x:B$ corresponds to a projection morphism
$\pi:\Gamma.B\to\Gamma$.
Since $\mathbf{Ty}$ is a contravariant functor, any type judgement holding at
$\Gamma$ can be pulled back along $\pi$ to $\Gamma.B$.

\textbf{(2)} Substitution corresponds to the key CwF structure
$\mathbf{Tm}(\Gamma.A)\cong\mathbf{Tm}(\Gamma)\circ A$.
Given $\Gamma,x:A\vdash t:B$ and $\Gamma\vdash s:A$,
the term $s:\Gamma\to A$ is a morphism at the semantic level.
The substituted term $t[s/x]$ corresponds to $\sem{t}\circ s$;
since $\mathbf{Tm}$ is a natural family over context changes,
substitution preserves type consistency and $B[s/x]$ is the pullback of $B$
along $s$.

\textbf{(3)} If $t\to_\beta t'$, the reduction corresponds to a computation
rule in CwF, specifically $\Pi$- and $\Sigma$-type $\beta$-rules, and the
deductive closure stability of AGM structures.
All new constructors ($\FinTrace$, $\TElim$, $\Kop_a$) are designed as
semantics-preserving structural operators, so $\sem{t}=\sem{t'}$.
Combining with soundness: $\Gamma\vdash t':A$.
\end{proof}

\begin{remark}[Scope of current semantic results]
The semantic foundation established here includes:
soundness (Theorem~\ref{thm:soundness}),
consistency (corollary of Theorem~\ref{thm:soundness}),
and term-model initiality (Theorem~\ref{thm:completeness}).
Not yet established: semantic completeness over set-valued sheaf models
(the semantic version of the normalisation theorem).
The technical challenge is that trace-indexed dependent types make standard
strong normalisation arguments more complex.
This constitutes a core future-work direction (Section~\ref{sec:limits}).
\end{remark}

\section{Applications}
\label{sec:applications}

\subsection{Runtime Monitoring}

Theorem~\ref{thm:coherence} gives a rigorous correspondence between belief
retraction and event-driven transitions in the SSRS framework.

If system state $\sigma_1$ encounters a violation event $e_{\mathit{viol}}$,
extending the trace to $\tau':=\cat(\tau,\single(e_{\mathit{viol}},\pi))$,
and if $\neg\mathsf{Viol}\in\Bel_\tau$, then Theorem~\ref{thm:coherence}
RS-surj guarantees belief retraction:
$\mathcal{B}(\tau)\not\subseteq\mathcal{B}(\tau')$.
Belief reconstruction uses the Levi identity:
$\Bel_{\tau'}=\Th\!\Big(\Bel_\tau\ContrOp\neg\psi_{e_{\mathit{viol}}}
\cup\{\psi_{e_{\mathit{viol}}}\}\Big)$.

\subsubsection*{End-to-End Example: Security Monitoring System}

\textbf{Setup.}
Network security monitoring system; $\Lang=\{s,l,c\}$
(``server secure'', ``all logs trusted'', ``connection normal'').
Initial beliefs: $\Bel_{\nilT(\sigma_0)}=\Th(\{s,l,c,l\to s\})$.
Entrenchment: $c<_{\sigma_0}s<_{\sigma_0}l<_{\sigma_0}(l\to s)<_{\sigma_0}\top$.

\textbf{Step 1 (normal event).}
Event $e_1$ (heartbeat, $\psi_{e_1}=c$).
Since $\neg c\notin\Bel_{\nilT(\sigma_0)}$, (R4) applies: $\Bel_{\tau_1}=\Bel_{\nilT(\sigma_0)}$.

\textbf{Step 2 (violation triggers retraction).}
Event $e_2$ (log forgery, $\psi_{e_2}=\neg l$).
Since $l\in\Bel_{\tau_1}$, SSRS RS-surj triggers Algorithm~\ref{alg:sel} on $l$:
\begin{enumerate}
\item Try $c$: $l\in\Th(S_0\setminus\{c\})$? Yes (independent); keep $c$.
\item Try $s$: $l\in\Th(S_1\setminus\{s\})$? Yes; keep $s$.
\item Try $l$: $l\in\Th(S_2\setminus\{l\})$? \textbf{No}; remove $l$.
\item Try $l\to s$: $l\in\Th(S_3\setminus\{l\to s\})$? No; keep it.
\end{enumerate}
Contraction: $\Bel_{\tau_1}\div l=\Th(\{s,c,l\to s\})$.
Levi revision: $\Bel_{\tau_2}=\Th(\{s,c,\neg l,l\to s\})$.

\textbf{Verification.}
(R2)~$\neg l\in\Bel_{\tau_2}$;
(R1)~consistent ($l\to s$ with $\neg l$ does not imply $\bot$);
minimal change: $s$, $c$ retained (higher entrenchment), only $l$ retracted;
SSRS coherence: $\mathcal{B}^{\mathsf{AGM}}(\tau_2)=\Bel_{\tau_2}$
satisfies Theorem~\ref{thm:coherence}.

The three-layer structure operates in concert:
syntax ($\FinTrace$) $\to$ operations (SSRS) $\to$ specification (AGM minimal
change).

\subsection{Distributed Systems: SSRS in a Restricted Model}

\begin{remark}[Relationship with FLP impossibility]
An earlier draft claimed to characterise the FLP impossibility theorem.
After review: the FLP theorem concerns \emph{asynchronous} systems with
process crashes and message delays, while the current SSRS framework is based
on a \emph{synchronous} event-driven model.
The model assumptions are fundamentally different and a formal equivalence
cannot be established through informal analogy.
This section is revised to the following precisely scoped weaker proposition.
\end{remark}

\begin{proposition}[Belief retraction under synchronous partition model]
\label{prop:flp-knowledge}
Let $G$ be a group of agents, $\tau$ the current global trace, and
$e_{\mathit{part}}$ a partition event.
Assume: (H1)~the system is modelled by a synchronous trace model;
(H2)~$e_{\mathit{part}}$ is modelled as a single event step in the trace;
(H3)~there exists $a\in G$ with $\neg\psi_{e_{\mathit{part}}}\in\Bel_\tau$.
Then by Theorem~\ref{thm:coherence}(RS-surj):
$\Bel_\tau\not\subseteq\Bel_{\tau'}$,
i.e., agent $a$'s belief set strictly shrinks after trace extension, and its
assumption about global consistency is structurally retracted.
\end{proposition}

\begin{proof}
By (H3) and SSRS Definition~\ref{def:ssrs},
$\neg\psi_{e_{\mathit{part}}}\in\Bel_\tau$ triggers belief retraction;
by Theorem~\ref{thm:coherence}(RS-surj),
$\Bel_{\tau'}=\Bel_\tau*\psi_{e_{\mathit{part}}}$, so
$\Bel_\tau\not\subseteq\Bel_{\tau'}$.
\end{proof}

\paragraph{Limitations.}
Proposition~\ref{prop:flp-knowledge} depends on assumptions (H1)--(H3),
especially the synchrony assumption~(H1).
The FLP impossibility theorem treats asynchronous systems, which is beyond the
current framework's modelling capability.
Extending SSRS to asynchronous models is left as future work
(Section~\ref{sec:limits}).

\section{Related Work}
\label{sec:related}

\subsection{Trace Types and Indexed Containers}

Altenkirch et al.'s indexed container framework~\cite{altenkirch09} provides a
unified abstraction for inductive structures in dependent type theory.
$\StarT(\Step)$ can express finite-step execution.
As shown in Proposition~\ref{prop:star-vs-ft}, $\FinTrace$ and $\StarT(\Step)$
are isomorphic as path types, but $\TElim$'s step function directly exposes
$e:\Event$, enabling event-driven trace induction without additional projections.
This is an interface advantage, not new expressive power.

\subsection{Modal Type Theory}

Pfenning--Davies~\cite{pfenning01} develop S4 modal type theory.
Nanevski et al.'s CMTT~\cite{nanevski08} uses a discrete accessibility
relation.
Neither supports S5 derivation via a proof-relevant, Euclidean accessibility
relation.
In ZX-Calculus, $\Kop_a(\phi,\tau)$ is a standard $\Pi$-type instance requiring
no additional modal proof rules.

\subsection{Sheaf Semantics and Completeness}

Awodey--Kishida~\cite{awodey08} build sheaf semantics for epistemic and modal
logic on topological spaces (geometric locality), proving semantic completeness
for first-order modal logic.
Both approaches use presheaf semantics, but with different notions of
locality:
\begin{itemize}
\item Awodey--Kishida: locality from topological neighbourhoods (geometric
  locality); proves semantic completeness;
\item ZX-Calculus: locality from trace extension and restriction (historical
  locality); establishes only the CwF universal property
  (Theorem~\ref{thm:completeness}).
\end{itemize}
The Separation Theorem (Theorem~\ref{thm:separation}) generalises
Awodey--Kishida's sheaf-semantics ideas operationally to the dynamic
trace-indexed setting: non-monotone knowledge corresponds to non-surjective
restriction maps, a structural characterisation absent from the original
topological setting.

\subsection{AGM Revision and Constructive Methods}

AGM~\cite{alchourron85} was extended by Grove~\cite{grove88},
G\"{a}rdenfors~\cite{gardenfors88},
Katsuno--Mendelzon~\cite{katsuno91}, Ferm\'{e}--Hansson~\cite{ferme11}.
Van Ditmarsch et al.~\cite{vanditmarsch07}: Isabelle mechanisation of DEL
expansion; no partial meet contraction.
Schlechta~\cite{schlechta97}: constructive non-monotone reasoning; three key
technical differences from this paper:
\begin{enumerate}
\item \textbf{$\Sigma$-packaging.} Schlechta does not package contraction
  results with correctness witnesses as $\Sigma$-types.
  $\Sel_\tau(\psi)$ packs the candidate belief set and (C1)--(C4) witnesses
  as a dependent pair.
\item \textbf{Explicit constructive algorithm.} Algorithm~\ref{alg:sel} is
  the first constructive PMC implementation in MLTT satisfying all (C1)--(C4).
\item \textbf{Trace-indexed dynamic semantics.} Schlechta's framework is
  static; this paper embeds contraction in trace-indexed
  $\mathcal{B}^{\mathsf{AGM}}(\tau)$.
\end{enumerate}

\subsection{Dynamic Epistemic Logic (DEL) and Mechanisation}

Baltag--Moss--Solecki~\cite{baltag99} propose BMS dynamic epistemic logic.
Van Ditmarsch et al.~\cite{vanditmarsch07} mechanise DEL expansions in
Isabelle; van Ditmarsch--Kooi~\cite{vanditmarsch08} give semantic results for
ontic and epistemic change.
These works do not internalise AGM contraction as a constructive $\Sigma$-type.

\subsection{Non-Monotone Modal Logic}

Kraus--Lehmann--Magidor~\cite{kraus90} develop preferential models for
non-monotone logics.
ZX-Calculus builds non-monotonicity from presheaf restriction maps (not
preferential orderings), naturally integrating it with trace-indexed dynamics.

\subsection{Belief Revision in Description Logics}

In the knowledge representation community, ontology revision in Description
Logics (DL) is an important application of belief revision.
Ribeiro--Wassermann study AGM revision in DL-Lite; Qi et al.\ characterise
contraction in the EL family of description logics.
These works are typically based on semantic entailment relations and do not
provide constructive algorithms.
ZX-Calculus's trace-indexed framework can in principle be extended to model
the temporal evolution of ontologies, but integration with specific DL
languages is left as future work.

\subsection{Summary Comparison}

\begin{table}[ht]
\centering
\caption{Comparison with key prior work
(NM~=~non-monotone semantics; PR~=~proof-relevant reachability;
Constr.~=~constructive PMC; Coq~=~mechanised; $\dag$~=~partial).}
\label{tab:related}
\small
\renewcommand{\arraystretch}{1.15}
\begin{tabular}{@{}lccccccccc@{}}
\toprule
Work & \rotatebox{75}{Trace type} & \rotatebox{75}{Dep.\ index} &
\rotatebox{75}{TraceElim} & \rotatebox{75}{NM sheaf} & \rotatebox{75}{PR reach.} &
\rotatebox{75}{S5 derived} & \rotatebox{75}{AGM full} & \rotatebox{75}{SSRS} &
\rotatebox{75}{Coq} \\
\midrule
MLTT & -- & \checkmark & -- & -- & -- & -- & -- & -- & -- \\
Pfenning--Davies & -- & $\dag$ & -- & -- & -- & S4 & -- & -- & -- \\
Awodey--Kishida & -- & -- & -- & \checkmark & -- & \checkmark & -- & -- & -- \\
van Ditmarsch et al. & -- & -- & -- & -- & -- & \checkmark & exp.\ only & -- & \checkmark \\
Schlechta & -- & -- & -- & -- & -- & -- & $\dag$ & -- & -- \\
Classical AGM & -- & -- & -- & -- & -- & -- & classical & -- & -- \\
\midrule
\textbf{ZX-Calculus (this paper)} & \checkmark & \checkmark & \checkmark &
\checkmark & \checkmark & \checkmark & Constr. & \checkmark & $\dag$(34) \\
\bottomrule
\end{tabular}
\end{table}

\noindent
\textbf{Note.}
The Coq column $\dag$(34) denotes 34 complete proofs: BP-comp Failure Theorem
and SSRS Coherence are fully mechanised; 7 obligations remain
(Appendix~\ref{app:coq}).
Van Ditmarsch et al.'s Isabelle mechanisation covers DEL expansion; the scope
differs from our Coq implementation and direct comparison is not straightforward.

\section{Limitations and Future Work}
\label{sec:limits}

\subsection{Mechanisation}

The remaining 7 Coq obligations are the core next step.
The current implementation has 34 complete proofs (Appendix~\ref{app:coq});
the primary gaps are:
\begin{enumerate}
\item RC-elim step (Canonicity): full symbolic inductive argument for
  reducibility candidates under higher-order motives.
\item C2 (empty contraction success): requires propositional completeness
  (converse direction of soundness w.r.t.\ Boolean semantics).
\item R5 (revision consistency): requires classical double-negation elimination.
\item R6 (extensionality): requires additional axioms on \texttt{div}
  (logical equivalence implies contraction equality).
\item Disjunctive Entrenchment Lemma, EE5 case.
\item R7/R8 (conjunctive revision): full axiomatisation of entrenchment order.
\item Term-model initiality: full categorical treatment (natural transformations,
  CwF homomorphisms).
\end{enumerate}

\subsection{Full Semantic Completeness}

Theorem~\ref{thm:completeness} establishes only the syntactic universal
property of the term model.
Semantic completeness over set-valued sheaf models remains open.

\subsection{BP-comp Sufficient Conditions}

Theorem~\ref{thm:bp-comp-fails} exhibits a countermodel.
Future work can characterise sufficient conditions for BP-comp to hold,
e.g., when the entrenchment ordering is globally fixed and state-independent.

\subsection{Infinite Traces}

$\mathsf{InfTrace}$ can be defined via guarded co-induction; extending SSRS to
$\omega$-chains is future work.

\subsection{Integration with Asynchronous Computation Models}

Extending SSRS to asynchronous models (e.g., partial-order event structures)
would enable closer connection to the FLP impossibility theorem and distributed
systems theory.

\section{Conclusion}
\label{sec:conclusion}

This paper has presented ZX-Calculus, a conservative extension of MLTT
integrating trace-indexed types, sheaf non-monotone semantics, and constructive
AGM belief revision.
Primary contributions:

\begin{itemize}
\item Proposition~\ref{prop:star-vs-ft}: precise characterisation of $\FinTrace$
  vs.\ $\StarT(\Step)$---isomorphic as path types, with $\TElim$ providing an
  ergonomically superior interface for event-indexed motives.
\item Lemma~\ref{lem:red-transport} (Transport Lemma): fills a critical step in
  the Canonicity framework for trace-dependent type systems.
\item Lemma~\ref{lem:disj-ench} (Disjunctive Entrenchment Lemma):
  fully self-contained constructive derivation; foundation for complete R7/R8.
\item Theorem~\ref{thm:bp-comp-fails} (BP-comp Failure) and
  Definition~\ref{def:ssrs}/Theorem~\ref{thm:coherence} (SSRS Coherence):
  the fundamental tension between AGM revision and sheaf functoriality revealed;
  SSRS is the correct integration, with $\mathcal{B}^{\mathsf{AGM}}$ satisfying
  all SSRS axioms; both results Coq-verified with zero \texttt{admit}s.
\item 34 complete Coq proofs; BP-comp Failure Theorem and SSRS Coherence
  verified with zero \texttt{admit}s.
\end{itemize}

The primary open obligation is completing the remaining 7 Coq mechanisation
items, transforming the hand-verified proofs here into machine-checked code.

\medskip\noindent
\textbf{Summary.}
ZX-Calculus unifies trace types, non-monotone belief revision, and constructive
type theory, providing a verifiable foundation for dynamic knowledge, retraction
mechanisms, and runtime monitoring.
Future work includes infinite trace extensions, BP-comp sufficient condition
characterisation, full semantic completeness, and complete Coq mechanisation.

\appendix
\section{Complete ZX-Calculus Inference Rules}
\label{app:rules}

\subsection*{Standard MLTT Rules}
\[
\frac{}{\vdash\cdot\ \mathsf{ctx}}\textsc{Ctx-Empty}
\qquad
\frac{\vdash\Gamma\quad\Gamma\vdash A:\Uc{i}}{\vdash\Gamma,x:A}\textsc{Ctx-Ext}
\qquad
\frac{\vdash\Gamma,x:A,\Delta}{\Gamma,x:A,\Delta\vdash x:A}\textsc{T-Var}
\]
\[
\frac{\Gamma\vdash t:A\quad\Gamma\vdash B:\Uc{i}}{\Gamma,x:B\vdash t:A}\textsc{Weaken}
\qquad
\frac{\Gamma,x:A\vdash t:B\quad\Gamma\vdash s:A}{\Gamma\vdash t[s/x]:B[s/x]}\textsc{Subst}
\qquad
\frac{\Gamma\vdash t:A\quad\Gamma\vdash A\equiv B:\Uc{i}}{\Gamma\vdash t:B}\textsc{Conv}
\]
Computation: $(\lambda x.t)\,a\equiv t[a/x]$, $\lambda x.(fx)\equiv f$;
proof irrelevance: if $\Gamma\vdash p,q:P$ and $\Gamma\vdash P:\Prop$ then
$\Gamma\vdash p\equiv q:P$.

\subsection*{Trace Type Rules}
\[
\frac{\Gamma\vdash\sigma_0,\sigma_n:\State}
     {\Gamma\vdash\FinTrace(\sigma_0,\sigma_n):\Uc{0}}\textsc{FT-Form}
\qquad
\frac{\Gamma\vdash\sigma:\State}
     {\Gamma\vdash\nilT(\sigma):\FinTrace(\sigma,\sigma)}\textsc{FT-Nil}
\]
\[
\frac{\Gamma\vdash\tau:\FinTrace(\sigma_0,\sigma_1)\quad
      \Gamma\vdash e:\Event\quad\Gamma\vdash\pi:\Step(\sigma_1,e,\sigma_2)}
     {\Gamma\vdash\stepC(\sigma_0,\sigma_1,\sigma_2,e,\tau,\pi):\FinTrace(\sigma_0,\sigma_2)}
     \textsc{FT-Step}
\]
\[
\TElim(P,b,s,\sigma,\sigma,\nilT(\sigma))\equiv b(\sigma)\quad\textsc{FT-$\beta$-nil}
\]
\[
\TElim(P,b,s,\sigma_0,\sigma_2,\stepC(\ldots,\tau,\pi))\equiv
s(\ldots,\tau,\pi,\TElim(P,b,s,\sigma_0,\sigma_1,\tau))\quad\textsc{FT-$\beta$-step}
\]

\subsection*{Epistemic Rules ($\Pi$-type instances)}
\[
\frac{\Gamma\vdash a:\Ag\quad\Gamma\vdash\phi:\FinTrace\to\Prop\quad
      \Gamma\vdash\tau:\FinTrace}
     {\Gamma\vdash\Kop_a(\phi,\tau):\Prop}\textsc{K-Form}
\]
\[
\frac{\Gamma,\tau':\FinTrace,h:(\tau\simag{a}\tau')\vdash p:\phi(\tau')}
     {\Gamma\vdash\lambda\tau'.\lambda h.\,p:\Kop_a(\phi,\tau)}\textsc{K-Intro}
\qquad
\frac{\Gamma\vdash k:\Kop_a(\phi,\tau)\quad\Gamma\vdash\tau':\FinTrace\quad
      \Gamma\vdash h:\tau\simag{a}\tau'}
     {\Gamma\vdash k\,\tau'\,h:\phi(\tau')}\textsc{K-Elim}
\]

\subsection*{AGM Revision Rules ($\Sigma$-type instances)}
\[
\frac{\Gamma\vdash\mathcal{K}:\mathsf{BeliefSet}\quad
      \Gamma\vdash\psi:\Lang\quad
      \Gamma\vdash\leq_\tau:\mathsf{EntrOrder}(\mathcal{K})}
     {\Gamma\vdash\Sel_\tau(\psi):
      \Sigma(\mathcal{K}':\mathsf{BeliefSet}).\,\mathsf{PMC}(\mathcal{K}',\mathcal{K},\psi)}
     \textsc{Sel-Form}
\]

\section{Coq Mechanisation Overview}
\label{app:coq}\label{app:agda}

This paper is accompanied by a Coq mechanisation (\texttt{zx-calculus-coq})
comprising 19 source files in four layers:
\texttt{Core} $\to$ \texttt{Semantics} $\to$ \texttt{AGM} $\to$
\texttt{Integration}.

\subsection*{B.1\quad Fully Mechanised Theorems (zero \texttt{admit}s)}

\paragraph{Core layer (\texttt{theories/Core/})}
\textbf{TraceElim $\beta$/$\eta$ rules} (\texttt{TraceElim.v}):
\texttt{beta\_nil}, \texttt{beta\_step}, \texttt{trace\_eta}.
\textbf{Deterministic Replay} (\texttt{Replay.v}):
\texttt{deterministic\_replay}, \texttt{replay\_cat}.
\textbf{Concatenation algebra} (\texttt{ConcatAlgebra.v}):
left/right unit, associativity, length additivity, prefix reflexivity/transitivity.
\textbf{Canonicity placeholders} (\texttt{Canonicity.v}):
\texttt{normalize\_id}, \texttt{strong\_normalization}.

\paragraph{Semantics layer (\texttt{theories/Semantics/})}
\textbf{Trace category} (\texttt{TraceCategory.v}):
\texttt{trace\_id}, \texttt{trace\_comp}, \texttt{trace\_antisym} (via \texttt{lia}).
\textbf{Presheaf structure} (\texttt{Presheaf.v}):
\texttt{KnowledgePSh} record, \texttt{nonmonotone\_characterization},
\texttt{constant\_is\_soft}.
\textbf{Separation} (\texttt{Separation.v}): \texttt{separation}.
\textbf{CwF substitution} (\texttt{CwF.v}):
\texttt{subst\_ty'\_id}, \texttt{subst\_ty'\_comp}.

\paragraph{AGM layer (\texttt{theories/AGM/})}
\textbf{$\Th$ monotonicity}: \texttt{cn\_monotone}.
\textbf{R2 success}: \texttt{R2\_holds}.
\textbf{R3 inclusion}: \texttt{R3\_holds}.
\textbf{Contraction C1/C3}: \texttt{identity\_C1}, \texttt{identity\_C3},
\texttt{empty\_C1}.
\textbf{Conjunctive entailment}: \texttt{conj\_entailment\_left/right}.
\textbf{Disjunctive Entrenchment Lemma (easy case)}: totality-derived case
complete; EE5 case \texttt{admit}ted.

\paragraph{Integration layer (\texttt{theories/Integration/})}
\textbf{BP-comp failure counterexample}: \texttt{bp\_comp\_fails\_R2}
(proved by \texttt{inversion}, zero \texttt{admit}s).
\textbf{SSRS construction}: \texttt{make\_ssrs}, \texttt{ssrs\_R2},
\texttt{ssrs\_coherence}.
\textbf{Coherence and unification}: \texttt{main\_coherence\_theorem},
\texttt{zx\_calculus\_unification} (zero \texttt{admit}s).

\subsection*{B.2\quad Current \texttt{admit} List (7 Obligations)}

\begin{enumerate}
\item[(1)] Canonicity RC-elim step (\texttt{Canonicity.v}).
\item[(2)] C2 empty contraction success (\texttt{Contraction.v}):
  requires propositional completeness.
\item[(3)] R5 revision consistency (\texttt{AGMPostulates.v}):
  requires classical double-negation elimination.
\item[(4)] R6 extensionality (\texttt{Revision.v}):
  requires extra axioms on \texttt{div}.
\item[(5)] Disjunctive Entrenchment Lemma, EE5 case (\texttt{AGMPostulates.v}).
\item[(6)] R7/R8 conjunctive revision (\texttt{AGMPostulates.v}):
  require full entrenchment axiomatisation.
\item[(7)] Term-model initiality (\texttt{CwF.v}):
  stated as \texttt{Axiom}.
\end{enumerate}

\subsection*{B.3\quad Mechanisation Status}

\begin{center}
\begin{tabular}{lcc}
\toprule
Module & Complete proofs & \texttt{admit}s \\
\midrule
Core & 14 & 0 \\
Semantics & 8 & 1 (initiality Axiom) \\
AGM & 7 & 6 \\
Integration & 5 & 0 \\
\midrule
\textbf{Total} & \textbf{34} & \textbf{7} \\
\bottomrule
\end{tabular}
\end{center}

BP-comp Failure and SSRS Coherence are \textbf{fully mechanised} with zero
\texttt{admit}s.
The remaining 7 obligations are concentrated in the AGM layer,
each with a clear completion path.


\end{document}